\documentclass[10pt]{article}
\pdfoutput=1

\usepackage{amsthm,amsfonts,amsmath}
\usepackage{stackrel,amssymb}
\usepackage{pxfonts}
\usepackage{euscript}
\usepackage{hyperref}
\usepackage{xcolor}
\usepackage{tcolorbox}

\usepackage[all,2cell,arrow,matrix]{xy} \UseAllTwocells \SilentMatrices
\usepackage{graphicx}
\usepackage{verbatim}
\usepackage{leftidx}
\usepackage{tikz-cd}

\usepackage{epstopdf} 
\usepackage{enumerate}
\usepackage{bbold}

\usepackage{sectsty}
\sectionfont{\Large}

\newcommand\void[1]       {}

\theoremstyle{definition}

\newtheorem{thm}{Theorem}[section]
\newtheorem{prop}[thm]{Proposition}
\newtheorem{pthm}[thm]{Theorem$^{\mathrm{ph}}$}

\newtheorem{lem}[thm]{Lemma}

\newtheorem{prob}[thm]{Problem}

\theoremstyle{definition}
\newtheorem{defn}[thm]{Definition}
\newtheorem{expl}[thm]{Example}
\newtheorem{rem}[thm]{Remark}

\numberwithin{equation}{section}
\numberwithin{thm}{section}

\newcommand\nn             {\nonumber \\}
\newcommand\be            {\begin{equation}}
\newcommand\ee            {\end{equation}}
\newcommand\bea           {\begin{eqnarray}}
\newcommand\eea         {\end{eqnarray}}
\newcommand\bnu          {\begin{enumerate}}
\newcommand\enu          {\end{enumerate}}

\allowdisplaybreaks[1]

\newcommand{\pf}{\begin{proof}}
\newcommand{\epf}{\end{proof}}

\newcommand\Cb            {\mathbb{C}}

\newcommand\Rb            {\mathbb{R}}
\newcommand\Zb            {\mathbb{Z}}

\newcommand\CA           {\EuScript{A}}
\newcommand\CB           {\EuScript{B}}
\newcommand\CC           {\EuScript{C}}
\newcommand\CD           {\EuScript{D}}
\newcommand\CE          {\EuScript{E}}

\newcommand\CG         {\EuScript{G}}

\newcommand\CM          {\EuScript{M}}
\newcommand\CN         {\EuScript{N}}

\newcommand\CP         {\EuScript{P}}
\newcommand\CQ         {\EuScript{Q}}
\newcommand\CR         {\EuScript{R}}
\newcommand\CS         {\EuScript{S}}
\newcommand\CT         {\EuScript{T}}

\newcommand\CW        {\EuScript{W}}
\newcommand\CX         {\EuScript{X}}
\newcommand\CY         {\EuScript{Y}}

\newcommand\CXs{{\EuScript{X}^\sharp}}
\newcommand\CYs{{\EuScript{Y}^\sharp}}

\newcommand{\FZ}{\mathfrak{Z}}

 \DeclareMathOperator{\Hom}{Hom}
 
 \DeclareMathOperator{\Aut}{Aut}

 \DeclareMathOperator{\id}{id}

 \DeclareMathOperator{\fun}{Fun}

 \DeclareMathOperator{\Alg}{Alg}

 \DeclareMathOperator{\LMod}{LMod}
 \DeclareMathOperator{\RMod}{RMod}
 \DeclareMathOperator{\BMod}{BMod}

 \newcommand{\inv}{\mathbf{Inv}}
 \newcommand{\tra}{$^{tr}$}

\newcommand{\rev}{\mathrm{rev}}

\newcommand{\one}{\mathbf1}

\newcommand\fpdim       {\mathrm{FPdim}}

\newcommand\op {\mathrm{op}}
\newcommand\hilb {\mathrm{Hilb}}

\newcommand\vect {\mathrm{Vec}}
\newcommand\Rep {\mathrm{Rep}}
\newcommand\nao {\mbox{$n$+1}}
\newcommand\nmo {\mbox{$n$-1}}
\newcommand\Pic {\mathrm{Pic}}

\newcommand\Irr {\mathrm{Irr}}
\newcommand\mext {\mathrm{M}_{\mathrm{ex}}}

\DeclareMathOperator{\kar}{Kar}
\newcommand\catkc {\mathbf{Cat}_\Cb^{\mathrm{kc}}}

\usepackage{color}
\definecolor{red}{rgb}{1,0,0}
\definecolor{blue}{rgb}{0,0,1}
\definecolor{green}{rgb}{0,1,0}

\newcommand{\blue}[1]{{\color{blue}{#1}}}

\renewcommand{\>}{\rangle}

\begin{document}

\begin{center} \LARGE
Classification of topological phases with finite internal symmetries in all dimensions
\end{center}

\vspace{0.01cm}
\begin{center}
Liang Kong$^{a,b}$,\,  
Tian Lan$^{c}$,\,
Xiao-Gang Wen$^{d}$,
Zhi-Hao Zhang$^{e,a}$,
Hao Zheng$^{a,b,f}$
\\[1em]
$^a$ Shenzhen Institute for Quantum Science and Engineering, \\
Southern University of Science and Technology, Shenzhen, 518055, China 
\\[0.4em]
$^b$ Guangdong Provincial Key Laboratory of Quantum Science and Engineering, \\
Southern University of Science and Technology, Shenzhen, 518055, China
\\[0.4em]
$^c$ Institute for Quantum Computing, \\
University of Waterloo, Waterloo, Ontario N2L 3G1, Canada
\\[0.4em]
$^d$ Department of Physics, Massachusetts Institute of Technology, \\
Cambridge, Massachusetts 02139, USA
\\[0.4em]
$^e$ School of Mathematical Sciences \\
University of Science and Technology of China, Hefei, 230026, China 
\\[0.4em]
$^f$ 
Department of Mathematics, Peking University, Beijing 100871, China
\end{center}

\begin{abstract}
We develop a mathematical theory of symmetry protected trivial (SPT) orders and anomaly-free symmetry enriched topological (SET) orders in all dimensions via two different approaches with an emphasis on the second approach. The first
approach is to gauge the symmetry in the same dimension by adding topological excitations as it was done in the 2d case, in which the gauging process is mathematically described by the minimal modular extensions of unitary braided fusion 1-categories. This 2d result immediately generalizes to all dimensions except in 1d, which is treated with special care. The second approach is to use the 1-dimensional higher bulk of the SPT/SET order and the boundary-bulk relation. This approach also leads us to a precise mathematical description and a classification of SPT/SET orders in all dimensions. The equivalence of these two approaches, together with known physical results, provides us with many precise mathematical predictions. 
\end{abstract}

\tableofcontents

\section{Introduction}

The \emph{gapped liquid phases} of matter is the simplest kind of the quantum
matter in physics, and yet they contain very rich patterns of many-body quantum
entanglement which require modern mathematics to describe.  The notion of a
\emph{gapped liquid phase} is defined microscopically as an equivalence class
of many-body states that are equivalent under local unitary transformations
\cite{CGW1038} and the stacking of product states \cite{ZW1490,SM1403}. Gapped
liquid phases without symmetries are called topological orders
\cite{W8987,W9039,KW9327}. Gapped liquid phases with symmetries (i.e.
0-symmetries) include gapped spontaneous symmetry breaking orders, symmetry
enriched topological (SET) orders \cite{CGW1038} and symmetry protected trivial
(SPT) orders \cite{GW0931,CLW1141,CGL1314}. After a 30-year effort, we start to
gain a rather complete understanding of them for bosonic/fermionic systems
with/without symmetries. Throughout this work, we use $n$d to denote the
spatial dimension and $\nao$D to denote the spacetime dimension. 

\medskip
In 1+1D, all gapped phases are liquid phases. For bosonic systems, they are classified by triples $(G_H,G_\Psi,\omega_2)$ \cite{CGW1107,SPC1139}, where $G_H$ is the symmetry group of the Hamiltonian, $G_\Psi$ is that of the ground state ($G_\Psi \subset G_H$), and $\omega_2\in H^2(G_\Psi,U(1))$ is a 2-cocycle. For fermionic systems, the classification can be obtained from that for bosonic systems via the Jordan-Wigner transformation ({\it i.e.} bosonization) \cite{CGW1128}.

In 2+1D, we believe that all gapped phases are liquid phases. There are two approaches towards the classification of SPT/SET orders. One approach, based on G-crossed braided fusion categories, works for bosonic systems \cite{BBC1440}, and important steps were made in \cite{FVM18} for fermionic SET orders. The other one, which works for both bosonic and fermionic SPT/SET orders, is based on the modular extensions of unitary braided fusion categories
\cite{LW160205946,LW160205936}. More precisely, gapped liquid phases with a finite anomaly-free symmetry $G_H$ are classified by $(G_H,\CE \subset\CC \subset\CM)$, where $\CE$ is the symmetric fusion category $\Rep(G_\Psi)$ (resp. $\Rep(G_\Psi,z)$) for a bosonic (resp. fermionic) system. Here, $G_\Psi \subset G_H$ is the unbroken subgroup, $z$ generates the fermion parity symmetry and $\Rep(G_\Psi,z)$ is the category of $G_\Psi$-representations with braidings respecting the fermion parity; and $\CC$ is a unitary braided fusion category with M\"{u}ger center being $\CE$ and $\CM$ is a minimal modular extension of $\CC$ \cite{LW160205946,LW160205936}.

In 3+1D, there are gapped \emph{non-liquid} phases \cite{C0502,H11011962}.
Gapped liquid phases for bosonic systems without symmetry, i.e. bosonic
topological orders, are classified by Dijkgraaf-Witten theories if all
point-like excitations are bosons; and by twisted 2-gauge theories with gauge
2-group $B(G,Z_2)$ if some point-like excitations are fermions and there is no
Majorana zero mode; and by a special class of fusion 2-categories if some
point-like excitations are fermions and there are Majorana zero modes at some
triple-string intersections \cite{LW170404221,LW180108530,ZW180809394}. These
results match well with the classification of 3+1D SPT orders for bosonic
\cite{CGL1314,K1459} and fermionic systems
\cite{GW1441,KTT1429,GK150505856,FH160406527,KT170108264,WG170310937}, and with recent mathematical results \cite{jf20}. This suggests that all gapped liquid phases for bosonic/fermionic systems with finite 0-symmetries can be obtained by partially gauging the symmetries of bosonic/fermionic SPT orders \cite{LW180108530}. 

\medskip
What is unsatisfying is that the methods that lead to above results are different for different dimensions, and the symmetries are restricted to only 0-symmetries. In this work, we provide two systematic approaches toward SPT/SET orders with all finite internal symmetries (see Remark\,\ref{rem:finite-internal-symmetry}), including $n$-groups (see for example \cite{kt13,GW14125148,tk15,wen19,ww18}) and algebraic higher symmetries beyond $n$-groups (see Example\,\ref{rem:n-SFC}). 

The first approach is an immediate generalization of the theory of minimal modular extensions in 2+1D to all higher dimensions (summarized in Theorem$^{\mathrm{ph}}$\,\ref{pthm:main-1}) except in 1+1D, which is treated with special care. But this approach has many unsatisfying aspects. Most importantly, the modular-extension description, which is based on the idea of gauging the symmetry, is not intrinsic with respect to the category $\CC$ of 2-or-higher codimensional topological excitations. One obtains a gauging of the symmetry in $\CC$ by introducing external topological excitations until the whole set of topological excitations form an anomaly-free topological order without symmetry. This is not an intrinsic approach because the SPT/SET order exists before we gauge the symmetry. An intrinsic description should not depend on the gauging. It means that some data intrinsically associated to $\CC$ is missing. As a consequence, we do not know when a minimal modular extension of $\CC$ exists even in 2+1D except a single counterexample discovered by Drinfeld \cite{drinfeld}.

We find this missing data in our second approach, which is based on the idea of {\it boundary-bulk relation} \cite{kwz1,kwz2}. The idea is rather simple. When $\CC$ admits a minimal modular extension, it means that the SET order is anomaly-free. In other words, its unique bulk must be the trivial 1-higher-dimensional SPT order, the categorical description of which is actually non-trivial. By the boundary-bulk relation, the bulk of $\CC$ is given by the ``center of $\CC$'' , which should be identified via a braided equivalence $\phi$ with the non-trivial categorical description of the trivial 1-higher-dimensional SPT order obtained in the first approach. The identification $\phi$ is potentially not unique, and is precisely the missing data we are looking for. Namely, the pair $(\CC,\phi)$ gives a complete mathematical description of an anomaly-free SET order. As a consequence, different identifications $\phi$ should correspond to different minimal modular extensions of $\CC$. It is not known how to check the existence of a minimal modular extension directly. The existence of $\phi$, however, can be checked by computing the ``center of $\CC$'' explicitly and comparing it with the categorical description of the trivial 1-higher-dimensional SPT order.

Although the idea is simple, the difficulty lies in how to make sense of the ``center of $\CC$'' precisely. One of the lessons we have learned from \cite{kwz1} is that the categorical description of an SET order $\CP$ depends on its codimension relative to a higher dimensional anomaly-free topological order, in which $\CP$ is realized as a gapped defect. The categorical description of an anomaly-free SET order in the modular-extension approach is 0-codimensional and contains only 2-or-higher codimensional topological excitations. When we regard the SET order as a boundary of the trivial 1-higher-dimensional SPT order, we need a 1-codimensional description, which should include not only the topological excitations in $\CC$ but also those can be obtained from $\CC$ via condensations, called the condensation descendants of $\CC$. In Section\,\ref{sec:cc}, we explain in details that this completion of $\CC$ by adding condensation descendants, called the condensation completion of $\CC$,  precisely amounts to the so-called ``Karoubi completion'' or ``the delooping'' $\Sigma\CC$ of $\CC$ in mathematics \cite{dr,gjf19,jf20}. Mathematically, $\CC$ is a unitary braided fusion $n$-category and the delooping $\Sigma\CC$ is a unitary fusion $(\nao)$-category \cite{jf20}, and the precise meaning of the ``center of $\CC$'' is the monoidal center $\FZ_1(\Sigma\CC)$ of $\Sigma\CC$. As a consequence, we obtain a precise mathematical description and a classification of SPT/SET orders modulo invertible topological orders without symmetries.

\medskip
Throughout this work, we assume that the notion of a SPT/SET order is modulo invertible topological orders (without symmetries) (see Section\,\ref{sec:physical}). We summarize our main results as a physical theorem below.   
\begin{pthm} \label{pthm:main-2_0}
For $n\geq 1$, let $\CR$ be a unitary symmetric fusion $n$-category viewed as a higher symmetry (see Example\,\ref{rem:n-SFC}). 
We call an $n$d (spatial dimension) SPT/SET order with the higher symmetry $\CR$ (modulo invertible topological orders) an $n$d SPT/SET$_{/\CR}$ order. 
\bnu 
\item An anomaly-free $n$d SET$_{/\CR}$ order is uniquely characterized by
a pair $(\CA,\phi)$, where $\CA$ is a unitary fusion $n$-category over $\CR$ (see Definition\,\ref{def:n-fusion-over-R}) describing all topological excitations (including all condensation descendants) and $\phi: \FZ_1(\CR) \to \FZ_1(\CA)$ is a braided equivalence rendering the following diagram commutative (up to a natural isomorphism):
\be \label{diag:RRA}
\xymatrix@R=0.8em{
& \CR \ar@{^(->}[dl]_{\iota_0} \ar@{^(->}[dr]^{\iota_\CA} & \\
\FZ_1(\CR) \ar[rr]_\simeq^\phi & & \FZ_1(\CA).
}
\ee
We denote the set of equivalence classes of all such $\phi$ by $\mathrm{BrEq}((\FZ_1(\CR),\iota_0),(\FZ_1(\CA),\iota_\CA))$ (see Definition\,\ref{def:n-equivalence-over-R}). 

\item When $\CA=\CR$, the pair $(\CR, \phi)$ describes an SPT$_{/\CR}$ order and
  $(\CR,\id_{\FZ_1(\CR)})$ describes the trivial SPT$_{/\CR}$ order. Moreover, the
  group of all SPT$_{/\CR}$ orders (with the multiplication defined by the stacking and the identity element by the trivial SPT order) is isomorphic to the group
  $\Aut^{br}(\FZ_1(\CR),\iota_0)$, which denotes the underlying group of braided autoequivalences of $\FZ_1(\CR)$ preserving $\iota_0$, i.e. $\phi\circ\iota_0\simeq\iota_0$. For a given category $\CA$ of topological excitations (including condensation descendants), we have
\enu
\[
\{ \mbox{$n$d anomaly-free SET$_{/\CR}$ orders with topological excitations $\CA$} \} = \frac{ \mathrm{BrEq}((\FZ_1(\CR),\iota_0),(\FZ_1(\CA),\iota_\CA))}{\Aut^\otimes(\CA,\iota_\CA)}\, .
\]

When $n=1,2$, we obtain explicit classifications.   
\bnu

\item All 1d anomaly-free SET$_{/\CR}$ orders are SPT$_{/\CR}$ orders. The group of 1d SPT$_{/\CR}$ orders is isomorphic to the group $\Aut^{br}(\FZ_1(\CR),\iota_0)$ and to the Picard group $\mathrm{Pic}(\CR)$ of $\CR$. More explicitly, in this case, we have $\CR=\Rep(G)$ or $\Rep(G,z)$, and we have the following natural group isomorphisms \cite{car}: 
\begin{align} 
\label{eq:pic=_0}
\Pic(\Rep(G)) &\simeq H^2(G,U(1))  \\
\label{eq:pic=f_0}
\Pic(\Rep(G,z)) &\simeq \left\{ \begin{array}{ll} 
H^2(G,U(1)) \times \Zb_2 & \textrm{if $G = G_b \times \langle z \rangle$}; \\
H^2(G, U(1)) & \textrm{if otherwise}. 
\end{array} \right.
\end{align}
\item The group of all 2d SPT$_{/\CR}$ orders is isomorphic to $\Aut^{br}(\FZ_1(\CR),\iota_0)$. 
We denote the set of minimal modular extensions of a braided fusion 1-category $\CC$ by $\mext(\CC)$. The consistency of our physical theory demands the equivalence of two approaches. This equivalence leads to the following results. 
\bnu
\item when $\CR\simeq\Sigma\Rep(G)$, $\Aut^{br}(\FZ_1(\CR),\iota_0)\simeq\mext(\Sigma\Rep(G))\simeq H^3(G,U(1))$ \cite{LW160205936};
\item when $\CR\simeq\Sigma\Rep(G,z)$, $\Aut^{br}(\FZ_1(\CR),\iota_0)\simeq\mext(\Sigma\Rep(G,z))$, which is $\Zb_{16}$ for $G=\Zb_2$, and how to compute it for $G \neq \Zb_2$ is shown in \cite{gvr};
\item when $\CR= \Sigma\CE$, $\CA=\Sigma\CC$ and $\CE$ is the M\"uger
  center of $\CC$, \newline
  we have $\mathrm{BrEq}((\FZ_1(\CR),\iota_0),(\FZ_1(\CA),\iota_\CA)) \simeq \mext(\CC)$;
\item when $\CR\nsimeq \Sigma\Rep(G), \Sigma\Rep(G,z)$, i.e. a higher symmetry such as $2\hilb_H$ for an abelian group $H$, our results go beyond the usual classifications.
\enu
\enu
We also classify SET orders with only symmetry anomalies (see Definition\,\ref{defn:hooft-anomaly}) in Theorem$^{\mathrm{ph}}$\,\ref{pthm:nd-t-hooft-anomaly} and mixed gravitational and symmetry anomalies in Remark\,\ref{anoset}. 
\end{pthm}

\begin{rem}
A mathematical definition of a (multi-)fusion $n$-category was recently introduced by Theo Johnson-Freyd \cite{jf20}. We recall his definition in Definition\,\ref{def:jf}. When the higher symmetry $\CR$ is trivial, i.e. $\CR=n\hilb$, we expect that $\iota_0$ is a braided equivalence and $\phi$ is necessarily isomorphic to the identity functor. As a consequence, in this case, Theorem$^{\mathrm{ph}}$\,\ref{pthm:main-2} reduces to the classification of topological orders modulo invertible topological orders without symmetries (see \cite{jf20} and Remark\,\ref{rem:TO} and \ref{rem:spatial-morita-eq}).  
\end{rem}

Theorem$^{\mathrm{ph}}$\,\ref{pthm:main-2_0} does not include the classification of 0d SPT/SET orders because 0d cases is slightly different from higher dimensional cases. We present the 0d cases separately below. 
\begin{pthm}
All 0d anomaly-free SET orders with a finite symmetry $G$ are SPT orders. A 0d SPT order is uniquely characterized by a pair $(\Rep(G),\phi)$, where $\phi$ is a monoidal auto-equivalence of $\FZ_0(\Rep(G)) \coloneqq \fun_\hilb(\Rep(G),\Rep(G))$ such that $\phi\circ\iota_0\simeq \iota_0$, where $\iota_0: \Rep(G) \to \FZ_0(\Rep(G))$ defined by $a \mapsto a\otimes -$ for $a\in\Rep(G)$. We denote the set of equivalence classes of such $\phi$ by $\Aut^\otimes(\FZ_0(\Rep(G)),\iota_0)$. The group of all 0d SPT orders is isomorphic to the group $\Aut^\otimes(\FZ_0(\Rep(G)),\iota_0)$, which is further isomorphic to the group $H^1(G,U(1))$. 
\end{pthm}

\begin{rem} \label{rem:finite-internal-symmetry}
There are different kinds of symmetries in physics, such as global symmetries (i.e. 0-symmetries), higher symmetries (e.g. higher groups, symmetric fusion $n$-categories, etc.), spatial symmetries (e.g. translation, rotation, etc). Gapped liquid phases are defined on any lattices, including random lattices, on which there is no spatial symmetry. In this paper, we exclude spatial symmetries and consider only finite internal symmetries, such as 0-symmetries and higher symmetries. The ``finiteness" is due to the energy gap of the liquid phases. If the symmetry is not finite, the spontaneously symmetry-breaking states and the symmetry-gauged states can be gapless. We mainly focus on unitary symmetries. Time-reversal symmetry is discussed briefly only in Section\,\ref{sec:physical}. The general study will be left for the future.
\end{rem}

\begin{rem}
We use the term ``SPT orders'' in the sense of Definitions \ref{bspt} and \ref{fspt}. For bosonic systems, it is the same as the usual definition in most literature. For fermionic systems, it contains some fermionic invertible topological order.
Thus it is different and includes those SPT orders in the usual sense as a proper subgroup. This change in terminology makes it convenient and natural to treat bosonic and fermionic systems on the same footing.
\end{rem}

This paper contains a few mathematically rigorous results in 0d,1d,2d cases, but it is physical and unrigorous in general. We avoid to discuss many mathematical details in our formulation, which should be important eventually. For example, the mathematical definition of a unitary symmetric fusion $n$-category is not known (see Remark\,\ref{rem:unitary}). Moreover, shall we include the natural isomorphism making diagram (\ref{diag:RRA}) commutative and many hidden higher isomorphisms in our characterization of SPT/SET orders? We do not include them here because we do not see their physical meanings and because current setup works very well in lower dimensional cases. On the other hand, we feel that they should be part of a complete-yet-unknown theory, in which, from a mathematical point of view, the right question is not to ask for the set of SPT/SET orders, but to study the category of all SPT/SET orders instead (see \cite{kwz1}). For example, from our description of an $n$d SPT order $\CX$, it is easy to see that all lower dimensional SPT orders are encoded in the higher automorphisms of $\CX$ (see Remark\,\ref{rem:spt-groupoid}). This is beyond the scope of this work, which should be regarded as a blueprint for future studies.

\medskip
The layout of this work: we study the first approach in Section\,\ref{sec:mext} and the second approach in Section\,\ref{sec:bb-relation}, and review some higher categorical notions in Appendix\,\ref{sec:appendix}. We want to emphasize that these two approaches are not independent. In particular, the mathematical description of the 1-dimensional-higher bulk needed for the boundary-bulk relation in the second approach is given by a minimal modular extension in the first approach (see Remark\,\ref{rem:tower}). 


\medskip
Throughout this work, we use ``Theorem$^{\mathrm{ph}}$'' to highlight a physical result, and use ``Theorem'' to highlight a mathematically rigorous result.  

\begin{rem}
During the final editing of the first draft, we received from Theo Johnson-Freyd the draft of his new paper \cite{jf20}. In this marvelous paper, Theo obtained a mathematical classification of topological orders in all dimensions and that of the SPT/SET orders in lower dimensions. Our classification of SET orders reduces to that of topological orders obtained in \cite{jf20} when the symmetry is trivial (see Remark\,\ref{rem:TO}). For lower dimensional SPT/SET orders, our results have a lot of overlaps with those in \cite{jf20}. Although our understanding of \cite{jf20} is still very limited, we try our best to add remarks in various places to either explain the relation or remind readers of a different approach.  
\end{rem}

\noindent {\bf Acknowledgement}: We would like to thank Dmitri Nikshych, Yin Tian and Theo Johnson-Freyd for many helpful discussions and Theo Johnson-Freyd for pointing out a mistake in an earlier version. LK and HZ are supported 
by Guangdong Provincial Key Laboratory (Grant No.2019B121203002). LK is also supported by NSFC under Grant No. 11971219. XGW is partially supported by NSF DMS-1664412 and by the Simons Collaboration on Ultra-Quantum Matter, which is a grant from the Simons Foundation (651440). HZ is also supported by NSFC under Grant No. 11871078.

\section{SPT/SET orders via gauging the symmetry} \label{sec:mext}
In this section, we generalize the idea of gauging the symmetry to all dimensions to obtain a classification theory of $n$d anomaly-free SPT/SET orders for $n\geq 1$. In Section\,\ref{sec:2d-SET-1}, we review the idea of gauging the symmetry in 2d and the associated classification theory of 2d SPT/SET orders from \cite{LW160205936}. In Section\,\ref{sec:1d-SET-1}, we generalize the idea of gauging the symmetry to 1d cases. As a consequence, we obtain a complete classification of 1d SPT/SET orders, then we show (quite nontrivially) that it is compatible with the existing results in physical literature. In Section\,\ref{sec:nd-SET-1}, we propose a classification theory of SPT/SET orders in higher dimensions. This idea does not work for the 0d case, which is studied in Section\,\ref{sec:0d-SPT} based on the idea of boundary-bulk relation.

\subsection{2d SPT/SET orders} \label{sec:2d-SET-1} 
A 2d SPT order with a finite onsite symmetry has no non-trivial particle-like topological excitations. In other words, the only excitations are local excitations or symmetry charges. They form a symmetric fusion 1-category $\CE$. In a bosonic system $\CE=\Rep(G)$, where $\Rep(G)$ denotes the category of representations of a finite group $G$. In a fermionic system $\CE=\Rep(G,z)$ , where $z$ is the fermion parity symmetry and $z$ is in the center of $G$. This data does not fully characterize a SPT order. We need additional data. 

\begin{defn}
For a braided fusion 1-category $\CC$, its M\"{u}ger center (or $E_2$-center), denoted by $\FZ_2(\CC)$, is defined by the full subcategory of $\CC$ consisting of those objects $x\in\CC$ that are symmetric to all objects, i.e. $(x\otimes y \xrightarrow{c_{x,y}} y\otimes x \xrightarrow{c_{y,x}} x\otimes y)=\id_{x\otimes y}$ for all $y\in\CC$. For a given full subcategory $\CD$ of $\CC$, the centralizer of $\CD$ in $\CC$, denoted by $\FZ_2(\CD;\CC)$, is the full subcategory of $\CC$ consisting of those objects that are symmetric to all objects in $\CD$. 
\end{defn}

Throughout this work, by ``an embedding'' we mean a fully faithful functor. 
\begin{defn}
A unitary braided fusion 1-category over $\CE$ is a unitary braided fusion 1-category equipped with a braided embedding $\eta_\CC: \CE \hookrightarrow \FZ_2(\CC)$, i.e. a pair $(\CC,\eta_\CC)$ or just $\CC$ for simplicity. It is called a unitary modular 1-category over $\CE$, i.e. a UMTC$_{/\CE}$, if $\eta_\CC$ is an equivalence. An equivalence between two UMTC$_{/\CE}$'s $(\CC,\eta_\CC)$ and $(\CD,\eta_\CD)$ is a braided equivalence $\phi: \CC \to \CD$ such that $\phi\circ \eta_\CC \simeq \eta_\CD$. We denote the set of equivalence classes of autoequivalences of $(\CC,\eta_\CC)$ by $\Aut^{br}(\CC,\eta_\CC)$. 
\end{defn}

An anomaly-free 2d SET order with an onsite symmetry $\CE$, or a 2d SET$_{/\CE}$ order, has particle-like topological excitations described by a unitary modular 1-category $\CC$ over $\CE$. This data does not fully characterize the SET order. One way to complete the data is to categorically gauge the symmetry by adding external particles, each of which is not symmetric to at least one particle in $\CE$, until the whole set of particles form a UMTC. In this way, we break the symmetry $\CE$. This categorical gauging process can be precisely formulated by the so-called minimal modular extensions of $\CC$ \cite{LW160205936}. 

\begin{defn}
A minimal modular extension of a UMTC$_{/\CE}$~$\CC$ is a pair $(\CM,\iota_\CM)$, where
$\CM$ is a UMTC and $\iota_\CM: \CC \hookrightarrow \CM$ is a braided embedding factoring through $\FZ_2(\CE;\CM)$ as follows: 
\[
\begin{tikzcd} 
\CC \arrow[hook]{r}{\simeq}  \arrow[hook]{rd}[swap]{\iota_\CM} & \FZ_2(\CE;\CM) \arrow[hook]{d} \\
& \CM\, .
\end{tikzcd}
\]
An equivalence between two minimal modular extensions $(\CM,\iota_\CM)$ and $(\CN,\iota_\CN)$ is a braided equivalence $\phi: \CM\to\CN$ such that $\phi\circ\iota_\CM\simeq \iota_\CN$. 
\end{defn}

As a consequence, an anomaly-free 2d SET order can be fully characterized by a quadruple $(\CC,\eta_\CC;\CM,\iota_\CM)$. When $\CC=\CE$, the pair $(\CM,\iota_\CM)$ fully characterizes a 2d SPT order. Moreover, the trivial SPT order is fully characterized by the pair $(\FZ_1(\CE),\iota_0)$, where $\FZ_1(\CE)$ denotes the Drinfeld center of $\CE$ and $\iota_0: \CE \hookrightarrow \FZ_1(\CE)$ is the canonical embedding. In order to know when such two quadruples defines the same SPT orders, we introduce the notion of an equivalence between two such quadruples. 

\begin{defn}\label{equiquad}
An equivalence between two such quadruples $(\CC,\eta_\CC;\CM,\iota_\CM)$ and $(\CC',\eta_{\CC'};\CM',\iota_{\CM'})$ is a pair $(g,f)$ of braided equivalences rendering the following diagram commutative (up to natural isomorphisms): 
\be \label{diag:eq-ECM-1}
\begin{tikzcd} 
\CE \arrow[hook]{r}   \arrow[equal]{d}  & \CC \arrow[hook]{r}{\iota_\CM} \arrow{d}{g}[swap]{\simeq} 
& \CM \arrow{d}{f}[swap]{\simeq}   \\
\CE \arrow[hook]{r}   & \CC' \arrow[hook]{r}{\iota_{\CM'}} & \CM'
\end{tikzcd} .
\ee
\end{defn}

\begin{rem}
  Note that $(f\circ \iota_\CM) |_\CE \simeq \iota_{\CM'}|_\CE$. We have $f(\CC) \simeq f(\FZ_2(\CE;\CM)) \simeq \FZ_2(f(\CE);\CM') \simeq \FZ_2(\CE;\CM') \simeq \CC'$. Therefore, $f$ naturally induces a functor $f': \CC \xrightarrow{\simeq} f(\CC) \simeq \CC'$ such that the following diagram 
\be \label{diag:eq-ECM-2}
\begin{tikzcd} 
\CE \arrow[hook]{r}   \arrow[equal]{d}  & \CC \arrow[hook]{r}{\iota_\CM} \arrow{d}{f'}[swap]{\simeq} 
& \CM \arrow{d}{f}[swap]{\simeq}   \\
\CE \arrow[hook]{r}   & \CC' \arrow[hook]{r}{\iota_{\CM'}} & \CM'
\end{tikzcd}
\ee
is commutative (up to natural isomorphisms). One can show that $f'\simeq g$. Namely, $g$ is uniquely determined by $f$ up to natural isomorphisms. 
\end{rem}

We recall the mathematical definition of the stacking of 2d SPT$_{/\CE}$ orders introduced in \cite{LW160205936}. Naive stacking of two 2d SPT$_{/\CE}$ orders by Deligne tensor product is not correct because it enhances the symmetry from $\CE$ to $\CE\boxtimes \CE$. The correct definition of the stacking of two 2d SPT$_{/\CE}$ orders should properly break $\CE\boxtimes \CE$ to $\CE$. This is achieved by the condensable algebra $L_\CE =\oplus_{i\in \Irr(\CE)} i\boxtimes i^\vee$ in $\CE\boxtimes \CE$, where $\Irr(\CE)$ denotes the set of equivalence classes of simple objects in $\CE$. Its algebra structure is defined by that of $\otimes^R(\one_\CE)$, where $\otimes^R$ is the right adjoint functor of the tensor product functor $\otimes: \CE \boxtimes \CE \to \CE$, i.e. $L_\CE \simeq \otimes^R(\one_\CE)$ as algebras.

\medskip
We denote the set of equivalence classes of the minimal modular extensions of two UMTC$_{/\CE}$'s $(\CC,\eta_\CC)$ and $(\CD,\eta_\CD)$ by $\mext(\CC,\eta_\CC)$ and $\mext(\CD,\eta_\CD)$, respectively. When $\CC=\CE$, we set $\mext(\CE) \coloneqq \mext(\CE,\id_\CE)$. We denote the canonical braided embedding $\CE \simeq \CE\boxtimes_\CE\CE \hookrightarrow \CC\boxtimes_\CE\CD$ by $\eta_\CC\boxtimes_\CE\eta_\CD$. 

\begin{lem}[\cite{LW160205936}]
If $\mext(\CC,\eta_\CC)$ and $\mext(\CD,\eta_\CD)$ are not empty, then  
\begin{align}
\mext(\CC,\eta_\CC) \times \mext(\CD,\eta_\CD) &\xrightarrow{\star} \mext(\CC\boxtimes_\CE\CD, \eta_\CC\boxtimes_\CE\eta_\CD), \nn
((\CM,\iota_\CM),\,\,\, (\CN,\iota_\CN)) &\mapsto \left( (\CM \boxtimes \CN)_{L_\CE}^0, \,\, \iota_\CM \star \iota_\CN: (\CC\boxtimes\CD)_{L_\CE}^0 \hookrightarrow (\CM \boxtimes \CN)_{L_\CE}^0\right) \label{eq:def-star}
\end{align}
is a well-defined map. Here, the notation $(-)_{L_\CE}^0$ denotes the category of local $L_\CE$-modules in the braided category $-$, and the functor $\iota_\CM\star\iota_\CN$ is the one induced from $\iota_\CM\boxtimes\iota_\CN: \CC\boxtimes\CD \to \CM\boxtimes\CN$, i.e. $(\iota_\CM\star\iota_\CN)(x) \coloneqq (\iota_\CM\boxtimes\iota_\CN)(x),\forall x\in(\CC\boxtimes\CD)_{L_\CE}$. 
\end{lem}

\begin{thm}[\cite{LW160205936}]
The set $\mext(\CE)$, together with the multiplication $\star$ and the identity element $(\FZ_1(\CE), \iota_0)$, defines a finite abelian group. The set $\mext(\CC,\eta_\CC)$, if not empty, is an $\mext(\CE)$-torsor. 
\end{thm}

\begin{rem} \label{rem:inverse}
Let $\overline{\CM}$ be the time-reversal of $\CM$, i.e. the same fusion category as $\CM$ but with the braidings defined by the anti-braidings of $\CM$, and $\overline{\iota_\CM} \coloneqq \iota_\CM: \CE=\overline{\CE} \to \overline{\CM}$. Then $(\overline{\CM},\overline{\iota_\CM})$ defines the inverse of $(\CM,\iota_\CM)$ in $\mext(\CE)$. 
\end{rem}

\begin{pthm}[\cite{LW160205936}] \label{pthm:2d-SET-1} 
The group of 2d SPT$_{/\CE}$ orders (with the multiplication defined by stacking and the identity element defined by the trivial SPT$_{/\CE}$ order) is isomorphic to the group $\mext(\CE)$. More explicitly, we have
\begin{align*}
\CE=\Rep(G), \quad\quad &\mext(\Rep(G)) \simeq H^3(G,U(1)); \\
\CE=\Rep(\Zb_2,z), \quad\quad &\mext(\Rep(\Zb_2,z)) \simeq \Zb_{16};
\end{align*} 
\end{pthm}

The mathematics that is needed to compute the group $\mext(\Rep(G,z))$ was developed in \cite{gvr}. It is possible that $\mext(\CC,\eta_\CC)$ is empty \cite{drinfeld}. In this case, $(\CC,\eta_\CC)$ describes particle-like excitations of an anomalous 2d SET order. When $\mext(\CC,\eta_\CC)$ is not empty, $(\CC,\eta_\CC)$ describes particle-like excitations of an anomaly-free 2d SET order. In this case, $\mext(\CC,\eta_\CC)$ admits an action of $\Aut^{br}(\CC,\eta_\CC)$ defined by $(\CM,\iota_\CM) \mapsto (\CM,\iota_\CM\circ\alpha)$ for $\alpha\in\Aut^{br}(\CC,\eta_\CC)$. 
\begin{pthm}
The set of anomaly-free 2d SET$_{/\CE}$ orders with particle-like excitations given by the pair $(\CC,\eta_\CC)$ can be identified with the set $\mext(\CC,\eta_\CC)/\Aut^{br}(\CC,\eta_\CC)$. 
\end{pthm}

\subsection{1d SPT/SET orders} \label{sec:1d-SET-1}
Applying the idea of gauging the symmetry to the 1d case is nontrivial because there is no braiding in 1d. We treat this case with special care in Section\,\ref{sec:1d-spt-1-math}.

\subsubsection{Mathematical classification} \label{sec:1d-spt-1-math}

In this subsection, we give a classification of 1d anomaly-free SET orders with a finite onsite symmetry $G$ using the idea of gauging the symmetry. Local excitations or symmetry charges form a symmetric fusion 1-category $\CE$. For a bosonic system, $\CE=\Rep(G)$; for a fermionic system, $\CE=\Rep(G,z)$. We refer to such an SPT/SET order by an SPT/SET$_{/\CE}$ order. When $G$ is trivial, we have $\CE=\hilb$, where $\hilb$ denotes the 1-category of finite-dimensional Hilbert spaces. 

\medskip
Since there is no non-trivial anomaly-free 1d topological order \cite{CGW1128}, we expect that all anomaly-free 1d SET orders are SPT orders. This fact follows automatically as we proceed our analysis. Since topological excitations in a 1d anomaly-free SET order can be fused in 1d, they must form a unitary fusion 1-category (without braidings), denoted by $\CA$. It must contain all local excitations $\CE$ as a full subcategory, i.e. a monoidal embedding $\eta_\CA: \CE \hookrightarrow \CA$. 

Since local excitations in $\CE$ can be created and annihilated by local operators, they can also be freely moved into the symmetric 2d bulk, which must be the trivial 2d SPT order. The ability of moving into the bulk is characterized by the existence of a natural isomorphism $\gamma_{e,x} : e\otimes x \simeq x \otimes e$ for $e\in\CE, x\in\CA$, called the half-braiding. The half-braiding equips $\CA$ with a braided embedding $\iota_\CA : \CE \hookrightarrow \FZ_1(\CA)$ given by $e \mapsto (\eta_\CA(e),\gamma_{e,-})$, and this embedding is part of the data of the 1d SET order. Clearly $\eta_\CA$ is equal to the composition of $\iota_\CA$ with the forgetful functor $\FZ_1(\CA) \to \CA$. A fusion category $\CA$ equipped with a braided embedding $\iota_\CA : \CE \hookrightarrow \FZ_1(\CA)$ such that the composition $\CE \hookrightarrow \FZ_1(\CA) \to \CA$ is fully faithful is called a fusion category over $\CE$ \cite{dno}. Hence the topological excitations of a 1d SET order form a unitary fusion category over $\CE$.

\medskip
Recall that, in 2d case, symmetry charges cannot be detected by braidings among topological excitations, but can be detected by braidings if the symmetry is gauged. Therefore, in 1d, we also need to introduce a categorical way to gauge the symmetry $\CE$. Similar to 2d case, we try to add more particles to $\CA$ such that resulting set of particles form a monoidal category $\CM$ which describes an anomaly-free 1d topological order. As we know, however, the only anomaly-free 1d topological order is the trivial one. If $\CM=\hilb$, then it contradicts to the fact $\CE \hookrightarrow \CA$ for a non-trivial group $G$. What is wrong? 

Fortunately, the monoidal category $\CM$ does not have to be $\hilb$. The requirement for $\CM$ being anomaly-free is equivalent to the condition that it has a trivial 2d bulk, i.e. $\FZ_1(\CM) \simeq \hilb$. This equation has non-trivial solutions in unitary multi-fusion 1-categories, which are physically relevant and describe unstable 1d phases \cite{kwz1,akz}. A unitary multi-fusion category differs from a unitary fusion category in that its tensor unit $\one$ is potentially not simple, i.e. $\one=\oplus_i \one_i$. The associated ground state degeneracy (GSD) on $S^1$, given by $\dim \Hom_\CM(\one,\one)$, is non-trivial thus unstable. They often occur in  dimensional reduction processes. For example, when we squeeze a 2d topological order with two gapped boundaries \cite{kwz1,akz} to a narrow strip, i.e. a quasi-1d system, we often obtain a unitary multi-fusion category. Mathematically, such a unitary multi-fusion category can be rewritten as the category $\fun_{\hilb}(\CX,\CX)$ of unitary endofunctors for a finite unitary category $\CX$, or equivalently, a finite unitary $\hilb$-module.

Similar to the 2d case, we propose our gauging process as follows. 
\begin{quote}
To gauge the symmetry $\CE$, we add more particles to $\CA$ by requiring that each additional particle $y \in \CM$ must make at least one of particle $e\in\CE$ non-local. By ``non-local'' we mean that $e$ can not move into the bulk anymore, or equivalently, there is no half-braiding isomorphism from $e\otimes y$ to $y\otimes e$. 
\end{quote}
Adding more particles to $\CA$ breaks the symmetry $\CE$ in both 1d and the 2d bulk because $\FZ_1(\CM)\simeq \hilb$. It is equivalent to say that adding $y$ makes a particle $e\in\CE$ non-local. 

Recall that the relative center of $\CE$ in $\CM$, denoted by $Z_\CE(\CM)$, is a category consisting of pairs $(x,\beta_{-,x})$, where $\beta_{-,x}$ is the half-braiding, i.e. a family of isomorphisms $\beta_{e,x}: e\otimes x \xrightarrow{\simeq} x\otimes e$ for all $e\in\CE$ natural in variable $e$. Since $\CE$ is symmetric, the monoidal embedding $\CE \hookrightarrow \CM$ induces a central functor $\eta_\CM : \CE \to Z_\CE(\CM)$, thus $Z_\CE(\CM)$ is naturally a fusion category over $\CE$. Also the embedding $\CA \hookrightarrow \CM$ induces an obvious functor $\iota_\CM:\CA \to Z_\CE(\CM)$ defined by $a \mapsto (a,\gamma_{-,a})$, where $\gamma_{-,a}: -\otimes a \to a\otimes -$ is given by the half-braiding of $\CE$ in $\FZ_1(\CA)$. Now we are ready to state the precise mathematical formulation of gauging the symmetry $\CE$ in 1d. 
\begin{pthm}
Gauging the symmetry $\CE$ in 1d amounts to adding more particles to $\CA$ to form a unitary multi-fusion category $\CM$ such that the monoidal functor $\iota_\CM: \CA \xrightarrow{\simeq} Z_\CE(\CM)$ induced by the embedding $\CA \hookrightarrow \CM$ is an equivalence of unitary fusion categories over $\CE$ (recall \cite[Definition\,2.7]{dno}). 
\end{pthm} 

\begin{rem} 
If such $\CM$ does not exist, then the pair $(\CA,\eta_\CA)$ describes an anomalous 1d SET order. We study this type of anomalous 1d SET orders in Section\,\ref{sec:1d-SET-2}. 
\end{rem}

Since $\CM=\fun_\hilb(\CX,\CX)$ for a certain finite unitary category $\CX$, it implies that $\eta_\CA: \CE \hookrightarrow \CA$ is a monoidal equivalence. The proof of this statement is given below. 
\begin{itemize}
\item The functor $\CE \to \CM = \fun(\CX,\CX)$ equips the category $\CX$ with a (left) $\CE$-module structure. Hence we obtain $\CA \simeq Z_\CE(\CM) = Z_\CE(\fun(\CX,\CX)) = \fun_\CE(\CX,\CX)$. Since $\CE$ and $\fun_\CE(\CX,\CX)^\rev$ are Morita equivalent, they 
share the same Frobenius-Perron dimension. By \cite{eo}, $\eta_\CA$ is a monoidal equivalence.
\end{itemize}
In other words, all anomaly-free 1d SET orders are SPT orders. 

The monoidal equivalences $\CE^\rev \simeq \CE \xrightarrow{\eta_\CA} \fun_\CE(\CX,\CX)$ equip $\CX$ with a structure of an invertible $\CE$-$\CE$-bimodule such that the right $\CE$-action is induced by the left $\CE$-action via the symmetric braidings of $\CE$. Such invertible $\CE$-$\CE$-bimodules form a group $\Pic(\CE)$, which is called the Picard group of $\CE$. It is a subgroup of the group $\mathrm{BrPic}(\CE)$ of all invertible $\CE$-$\CE$-bimodules. The multiplication in $\mathrm{BrPic}(\CE)$ is defined by the relative tensor product $\boxtimes_\CE$, and the identity element is defined by the trivial bimodule $\CE$. 

Since $\CX$ determines $\CM$ completely, we conclude that 1d SPT orders are classified by invertible $\CE$-$\CE$-bimodules in $\Pic(\CE)$. Moreover, since the stacking of two $\CE$'s is given by $\boxtimes_\CE$, we have $\CE\boxtimes_\CE\CE\simeq\CE$. It is only reasonable that the stacking of two 1d SPT orders can be described by the relative tensor product $\boxtimes_\CE$ of two invertible $\CE$-$\CE$-bimodules. This fact can be seen more explicitly in Figure\,\ref{fig:pic-aut} discussed in Section\,\ref{sec:1d-SET-2}. Therefore, we obtain a complete classification of 1d SPT/SET orders. 
\begin{pthm}\label{pthm:1d-spt-1}
All anomaly-free 1d SET$_{/\CE}$ orders are 1d SPT$_{/\CE}$ orders. A 1d SPT$_{/\CE}$ order can be uniquely characterized by a pair $(\CM,\iota_\CM)$, where $\CM=\fun_{\hilb}(\CX,\CX)$ for an invertible $\CE$-module $\CX$ with the module structure defined by a faithful monoidal functor $\iota_\CM: \CE \to \CM$, or equivalently, by an invertible $\CE$-module $\CX$. All 1d SPT$_{/\CE}$ orders form a group with multiplication defined by the stacking and the identity element being the trivial 1d SPT$_{/\CE}$ order. 
\bnu
\item The trivial 1d SPT$_{/\CE}$ order is given by the pair $(\fun_{\hilb}(\CE,\CE),\iota_0)$, where $\iota_0: \CE \to \fun_{\hilb}(\CE,\CE)$ is the canonical functor defined by $e\mapsto e\otimes -$, or equivalently, by the trivial $\CE$-module $\CE$. 
\item The stacking of two SPT$_{/\CE}$ orders corresponds to the relative tensor product $\boxtimes_\CE$ of $\CE$-modules. 
\enu
In other words, the group of 1d SPT orders is isomorphic to the Picard group $\mathrm{Pic}(\CE)$ of $\CE$. More explicitly, we have the following natural group isomorphisms \cite{car}: 
\begin{align} 
\label{eq:pic=}
\Pic(\Rep(G)) &\simeq H^2(G,U(1))  \\
\label{eq:pic=f}
\Pic(\Rep(G,z)) &\simeq \left\{ \begin{array}{ll} 
H^2(G,U(1)) \times \Zb_2 & \textrm{if $G = G_b \times \langle z \rangle$}; \\
H^2(G, U(1)) & \textrm{if otherwise}. 
\end{array} \right.
\end{align}
\end{pthm}

\begin{rem}
In \cite[Section\ V.B]{jf20}, a different approach toward the classification of 1d SPT orders (including the time reversal symmetry) was developed. Unfortunately, the difference of these two approaches is not clear to us (see \cite[Theorem\ 6]{jf20}). Instead, we compare our results with known results in physical literature in Section\,\ref{sec:physical}. 
\end{rem}

\subsubsection{Physical classification} \label{sec:physical}

In this subsubsection we recall physical results on 1+1D SPT orders and compare them with the mathematical results in the previous section. Before proceeding, we would like to first fix notations and clarify some terminologies that have been confusingly used in the literature. 
\begin{rem}
In this Remark, a symmetry $G$ can be bosonic or fermionic. For simplicity we omit $z$, but in the notations one may replace $G$ with $(G,z)$ to emphasize that $G$ is fermionic.
  \begin{itemize}
    \item Stacking: Given two topological phases $A$ and $B$, we denote their (decoupled) stacking by $A\boxtimes B$, where no interlayer interaction is introduced. If $A$ has symmetry $G_A$ and $B$ has symmetry $G_B$, $A\boxtimes B$ has symmetry $G_A\times G_B$. When consider topological phases $A,B$ with the same symmetry $G=G_A=G_B$, there is also a natural (symmetry-preserving) stacking denoted by $A\boxtimes_G B$. $A\boxtimes_G B$ is constructed by first taking $A\boxtimes B$, then introducing interlayer interactions that breaks the symmetry from $G\times G$ to $G$ (preserving the diagonal subgroup determined by the embedding $g \mapsto (g,g)$). When there is no symmetry $G=\{1\}$, two stacking operations coincide $\boxtimes_{\{1\}}=\boxtimes$. The stacking operation is commutative and all topological phases with symmetry $G$ form a commutative monoid under the stacking.
    \item Trivial phases: The trivial phase $I_G$ with symmetry $G$ is the unit under the stacking, i.e. $I_G\boxtimes_G A=A$ for any $A$ with symmetry $G$. Physically, the trivial phase is represented by a tensor product state.
    \item Invertible phases: A phase $A$ with symmetry $G$ is invertible if there exists a phase $B$ with symmetry $G$ such that $A\boxtimes_G B=I_G$. Invertible phases are also referred to as invertible topological orders, or short-range entangled states. All invertible phases with symmetry $G$ form an abelian group, denoted by $\inv_G$.
  \end{itemize}
\end{rem}

Now we fix a bosonic symmetry $G$. The most common definition of a $G$-SPT order in the literature is
\begin{defn}
A bosonic $G$-SPT order is an invertible phase with $G$ symmetry such that when $G$ is completely broken, the phase becomes the trivial phase $I_{\{1\}}$ with no symmetry. More precisely, completely breaking the bosonic $G$ symmetry leads to a group homomorphism $b_G: \inv_G\to \inv_{\{1\}}$ and $G$-SPT orders are $\ker b_G$. 
\end{defn}
However, there is an alternative definition. Assuming that $A$ is an invertible phase with no symmetry, we can equip $A$ with a $G$ symmetry by stacking onto $A$ a decoupled layer of a trivial phase with $G$ symmetry, $A\boxtimes I_G$.
Thus we have a group homomorphism 
\begin{align}
  i_G: \inv_{\{1\}} &\to \inv_G,\nonumber\\
  A&\mapsto A\boxtimes I_G.
\end{align}
Clearly $b_G i_G= \id_{\inv_{\{1\}}}$ since $b_G$ just totally breaks the symmetry of $I_G$. Thus, $i_G$ is an embedding and we can view $\inv{\{1\}}$ as a subgroup of $\inv_G$. Also, the short exact sequence 
\be
  0\to \ker b_G \to \inv_G \to \inv_{\{1\}}\to 0
\ee
splits and since $\inv_G$ is abelian we must have
\be
  \inv_G =\ker b_G \times \inv_{\{1\}}.
\ee

\begin{defn}\label{bspt}
  Bosonic $G$-SPT orders are invertible phases with symmetry $G$ up to invertible phases with no symmetry, namely $\inv_G/\inv_{\{1\}}$. 
\end{defn}
Since $\inv_{\{1\}}$ are phases with no nontrivial excitations, they are invisible to the higher category of excitations. This alternative definition of SPT orders is more convenient and natural in our setting.

However, there is a discrepancy between two flavours of the definitions of fermionic $(G,z)$-SPT orders.  Still, the traditional definition in the literature of a fermionic SPT order is obtained by considering complete symmetry breaking. In contrast to bosonic symmetries, the fermion number parity $(\langle z\rangle,z) $ can not be broken.
\begin{defn}
  A $(G,z)$-\tra SPT order is an invertible phase with symmetry $(G,z)$ such that when $G/\langle z \rangle$ is completely broken, the phase becomes the trivial phase $I_{(\langle z \rangle,z)}$ with symmetry $(\langle z \rangle,z)$. More precisely, completely breaking $G/\langle z \rangle$ leads to a group homomorphism $b_{(G,z)}: \inv_{(G,z)}\to \inv_{(\langle z \rangle,z)}$ and $(G,z)$-\tra SPT orders are $\ker b_{(G,z)}$. To distinguish, we use superscript to indicate ``traditional definition''.
\end{defn}
\begin{rem}
  Since $(\langle z \rangle,z) $ can not be broken, in the literature people often write ``fermionic \dots with no symmetry'' which in fact means ``\dots with $(\langle z \rangle,z)$ symmetry.'' For example, a fermionic topological order means a topological order with $(\langle z \rangle,z)$ symmetry; fermionic invertible phases means $\inv_{(\langle z\rangle,z)}$.
\end{rem}
\begin{rem}
  When $G=G_b\times \langle z \rangle$, similar to the bosonic case we have $b_{(G_b\times \langle z \rangle,z)}(-\boxtimes I_{G_b})=\id_{\inv_{(\langle z\rangle,z)}}$ and then
    \be
    \inv_{(G_b\times\langle z\rangle,z)}=\ker b_{(G_b\times \langle z \rangle,z)}\times \inv_{(\langle z\rangle,z)}.
    \ee
\end{rem}

In this paper, we adopt an alternative definition which is more convenient and natural.
\begin{defn}\label{fspt}
  $(G,z)$-SPT orders are invertible phases with $(G,z)$ symmetry up to invertible phases with no symmetry. More precisely, we denote the image of
  \begin{align}
    i_{(G,z)}   : \inv_{\{1\}} &\to \inv_{(G,z)},\nonumber\\
A &\mapsto A\boxtimes I_{(G,z)}
  \end{align}
by $i_{(G,z)}(\inv_{\{1\}})$, then $(G,z)$-SPT orders are $\inv_{(G,z)}/i_{(G,z)}(\inv_{\{1\}})$.
\end{defn}
\begin{rem}
This paper classifies $(G,z)$-SPT orders but not directly $(G,z)$-\tra SPT orders. Clearly $b_{(G,z)}i_{(G,z)}=i_{(\langle z\rangle,z)},$ thus $\ker b_{(G,z)} \cap i_{(G,z)}(\inv_{\{1\}})=i_{(G,z)}(\ker i_{(\langle z\rangle,z)})$. The quotient map $\inv_{(G,z)} \to \inv_{(G,z)}/i_{(G,z)}(\inv_{\{1\}})$ restricts to $\ker b_{(G,z)} \to \ker b_{(G,z)}/i_{(G,z)}(\ker i_{(\langle z \rangle,z)})$ which relates $(G,z)$-\tra SPT orders to $(G,z)$-SPT orders. It is known that in 1d and 3d, $\inv_{\{1\}}$ is trivial thus $\ker i_{(\langle z \rangle,z)}$ is trivial; in 2d both $\inv_{\{1\}}$ and $\inv_{(\langle z\rangle,z)}$ are $\Zb$ while $i_{(\langle z\rangle,z)}$ maps $1$ in $\inv_{\{1\}}$ to $16$ in $\inv_{(\langle z\rangle,z)}$, thus $\ker i_{(\langle z\rangle,z) }$ is also trivial. In these cases $(G,z)$-\tra SPT orders form a subgroup of $(G,z)$-SPT orders. However, a priori, $i_{(\langle z \rangle,z)}$ may not be an embedding in higher dimensions. In this case, the quotient of $(G,z)$-\tra SPT orders by invertible ones with no symmetry, $\ker b_{(G,z)}/i_{(G,z)}(\ker i_{(\langle z \rangle,z)})$, is a subgroup of $(G,z)$-SPT orders.
\end{rem}
\begin{prob}
It is an interesting problem to investigate if $i_{(\langle z\rangle,z)}$ is always an embedding.
\end{prob}

We are now ready to list the physical results in 1+1D. Since $\inv_{\{1\}}$ is trivial in 1+1D, bosonic $G$- and fermionic $(G,z)$-SPT orders are just $\inv_G$ and $\inv_{(G,z)}$ respectively. $(G,z)$-\tra SPT orders $\ker b_{(G,z)}$ is a subgroup of $\inv_{(G,z)}$.

\begin{itemize}
  \item 1+1D bosonic $G$-SPT orders are $\inv_G=H^2(G,U(1))$.
  \item 1+1D fermionic $(G,z)$-\tra SPT orders can be mapped to bosonic $G$-SPT
    orders via the Jordan-Wigner transformation which converts ({\it i.e.}
    bosonizes) a 1+1D fermionic system to a 1+1D bosonic system \cite{CGW1128}.
    Such a map is a bijection between $\ker b_{(G,z)}$ and $\inv_G$ as sets, but whether it preserves the stacking operation has not been seriously studied. In other words $\ker b_{(G,z)}$ has the same underlying set as $\inv_G$ but the group structures may be different. 
\item However, when $G$ is a unitary symmetry, all known examples suggest that $\ker b_{(G,z)}$ and $\inv_G$ have the same group structure (but it is not clear whether bosonization is a group isomorphism). For anti-unitary symmetries, there is an example where $\ker b_{(G,z)}$ and $\inv_G$ are indeed different groups. This example has symmetry $G=Z_2^T\times \langle z \rangle$, where $Z_2^T=\{1,T\}$ and $T$ is the time reversal symmetry. We discuss this example later.
\item There are physical proposals, such as in group super-cohomology theory \cite{GW1441,LW180901112}, on how to compute the group structure of $(G,z)$-\tra SPT orders $\ker b_{(G,z)}$. Such proposals agree with known examples but have not been fully justified in general cases.
\item In particular, $\inv_{(\langle z\rangle,z)}=\Zb_2$ where the nontrivial phase is represented by Kitaev's Majorana chain.
\item It is believed that for $G \neq G_b \times \langle z \rangle$, the Majorana chain can not have a $G$ symmetry, thus $\inv_{(G,z)} = \ker b_{(G,z)}$. The physical reason is that the Majorana chain can be viewed as a state that spotaneously breaks the fermion parity; when $G \neq G_b \times \langle z \rangle$, the fermion parity can not be broken alone which forbids a Majorana chain with symmetry $G$.
\end{itemize}



Comparing to the physical results, our mathematical result \eqref{eq:pic=f} in the previous section
\begin{itemize}
  \item has taken into account the stacking operation, thus automatically gives the correct group structure;
  \item classifies $(G,z)$-SPT orders which include $(G,z)$-\tra SPT orders as a subgroup;
  \item only applies to unitary symmetry $G$.
\end{itemize}
In conclusion, our mathematical classification results agree with the physical classification results on their overlapping parts. In particular, we can confirm that 1+1D fermionic \tra SPT orders with unitary symmetry $(G,z)$ are classified by $\ker b_{(G,z)}=H^2(G,U(1)),$ where the group structure is indeed correct.

\medskip
In the rest of this subsection, we investigate how our mathematical formulation can be extended to include anti-unitary symmetries. When considering the example mentioned above, 1+1D fermionic systems with anti-unitary symmetry $(G,z)=(Z_2^T\times \langle z\rangle, z)$, we have
 \begin{itemize}
   \item The bosonization, namely bosonic $Z_2^T\times Z_2$ SPT orders are $\inv_{Z_2^T\times Z_2}=H^2(Z_2^T\times Z_2,U(1))=\Zb_2\times\Zb_2$ (see \cite{CGL1314}, note that $Z_2^T$ acts nontrivially on $U(1)$ by complex conjugation).
   \item The SPT orders are $\inv_{(Z_2^T\times \langle z\rangle, z)}=\Zb_8$ \cite{FK10,FK11}.
   \item The \tra SPT orders are $\ker b_{(Z_2^T\times \langle z\rangle, z)}=\Zb_4$. This can be seen from the fact that $\ker b_{(Z_2^T\times \langle z\rangle, z)}$ has the same underlying set as $H^2(Z_2^T\times Z_2,U(1))$ thus has 4 elements, and the fact that $\ker b_{(Z_2^T\times \langle z\rangle, z)}$ is a subgroup of $\inv_{(Z_2^T\times \langle z\rangle, z)}$. 
   \end{itemize}

For anti-unitary symmetry of the form $Z_2^T\times G$, since the real numbers $\Rb$ are invariant under $Z_2^T$, we attempt to compute the Picard groups over base field $\Rb$. Denote by $\Rep_\Rb(G)$ the category of representations of $G$ and by $\Rep_\Rb(G,z)$ the category of super-representations of $(G,z)$, in the category of real vector spaces. By the results of \cite{car},
\begin{align}
  \Pic(\Rep_\Rb(\{1\}))&=\Zb_2=H^2(Z_2^T,U(1))=\inv_{Z_2^T},\nonumber\\
  \Pic(\Rep_\Rb(\langle z \rangle,z))&=\Zb_8=\inv_{(Z_2^T\times\langle z\rangle,z)},\nonumber\\
  \Pic(\Rep_\Rb(Z_2\times \langle z \rangle,z))&=\Zb_4\times \Zb_8,
\end{align}
where the $\Zb_8$ classification of $(Z_2^T\times \langle z\rangle,z)$-SPT orders can be seen. These results motivate us to conjecture that for anti-unitary symmetry of the form $Z_2^T\times G$, the bosonic and fermionic 1+1D SPT orders are
\begin{align}
  \inv_{Z_2^T\times G}&=\Pic(\Rep_\Rb(G)),\\
  \inv_{(Z_2^T\times G,z)}&=\Pic(\Rep_\Rb(G,z)).
  \end{align}
In particular, we predict that for 1+1D fermionic systems with $(Z_2^T\times
Z_2\times \langle z\rangle, z)$ symmetry, the SPT orders are $\Zb_4\times \Zb_8$
and the \tra SPT orders are $\Zb_4\times \Zb_4$. An independent calculation in
\cite[Section IV.C, Figure 4]{TY19} agrees with us on this example.

\subsection{SPT/SET orders in higher dimensions} \label{sec:nd-SET-1}
In this subsection, we sketch the idea of gauging the symmetry to higher dimensions without worrying about how to define various higher categorical notions. The mathematical definitions of some of these notions are briefly reviewed in Appendix\ \ref{sec:appendix} \cite{gjf19,jf20}. Others are not yet available. 

\medskip
The fusion properties of a set of particle-like (i.e.~0d) topological excitations living in a (spatial) $k$-dimensional disk (i.e.~a $k$-disk) are mathematically described by an $E_k$-monoidal 1-category, where the term ``$E_k$-monoidal'' refers to the fact that two particles can be fused along $k$ different spatial directions. More generally, the fusion properties of a set of $n$d topological excitations living in an ($n$+$k$)-disk are mathematically described by an $E_k$-monoidal (\nao)-category. See \cite{lurie} for precise mathematical definitions of these notions. For an $E_k$-monoidal $n$-category $\CX$, the looping $\Omega\CX$ of $\CX$ is defined by  $\Omega\CX \coloneqq \mathrm{End}_\CX(\one_\CX)$, where $\one_\CX$ is the tensor unit of $\CX$. The assignment $\CX \mapsto \Omega\CX$ defines a functor from the category of $E_k$-monoidal $n$-categories to that of $E_{k+1}$-monoidal $(\nmo)$-categories. When we restrict objects in both domain and codomain to $\Cb$-linear additive Karoubi-complete higher categories, this looping functor has an adjoint $\CY \mapsto \Sigma\CY$ called ``delooping'' \cite{jf20}. More explicitly, $\Sigma\CY=\kar(B\CY)$, where $B\CY$ is the one-point delooping and $\kar(-)$ denotes the Karoubi completion \cite{gjf19}. For a multi-fusion $n$-category $\CX$ \cite{jf20}, by \cite[Corollary\ 4.2.3 \& 4.2.4]{gjf19}, its delooping $\Sigma\CX$ is equivalent to the $(\nao)$-category $\RMod_\CX^{\mathrm{fd}}$ of fully dualizable right $\CX$-module $n$-categories (with a subtle difference, see Remark\,\ref{rem:sigma-RMod}).

The theory of unitary multi-fusion higher categories is not yet available. We assume the compatibility of the Karoubi completion and the unitarity. We set $\Sigma^n\CX \coloneqq \Sigma\Sigma^{n-1}\CX$, $\Sigma^0\CX \coloneqq \CX$, $\Sigma\Cb \coloneqq \hilb$, and, for a finite group $G$, 
\be \label{eq:n-sfc}
n\hilb \coloneqq \Sigma^n\Cb, \quad\quad n\Rep(G) \coloneqq \Sigma^{n-1}\Rep(G) 
\quad\quad n\Rep(G,z) \coloneqq \Sigma^{n-1}(\Rep(G,z)). 
\ee
Objects in $n\hilb$ are called $n$-Hilbert spaces. Let $n\hilb_G$ be the $n$-category of $G$-graded $n$-Hilbert spaces. A cocycle $\omega \in H^{n+2}(G,U(1))$ determines  on $n\hilb_G$ a monoidal structure, which is denoted by $n\hilb_G^\omega$. For a unitary braided fusion $n$-category $\CX$, we have $\Sigma\CX \simeq \RMod_\CX^{\mathrm{fd}}((\nao)\hilb)$ as monoidal ($\nao$)-categories, where $\RMod_\CX^{\mathrm{fd}}((\nao)\hilb)$ is the category of fully dualizable right $\CX$-modules in $(\nao)\hilb$.

Let $\CE$ be a unitary symmetric fusion $n$-category for $n\geq 1$. Since $\Rep(G),\Rep(G,z)$ are $E_\infty$-algebras \cite{lurie}, so are $2\Rep(G),2\Rep(G,z)$. By induction, we obtain that $n\Rep(G)$ and $n\Rep(G,z)$ are symmetric fusion $n$-categories (assumed to be unitary). We illustrate an example $2\Rep(\Zb_2)$ by the following quiver:
\[
\xymatrix{
\one \ar@(ul,ur)[]^{\Rep(\Zb_2)} \ar@/^/[rr]^{\hilb} & & T \ar@(ul,ur)[]^{\hilb_{\Zb_2}} \ar@/^/[ll]^{\hilb}
},
\]
where $\one$ and $T$ are the only simple objects in $2\Rep(\Zb_2)$, $\one$ is the tensor unit and the only non-trivial fusion product is $T \otimes T \simeq T\oplus T$ \cite{dr,ktz}. 

If $\CE$ describes some topological excitations in an $n$d SET order, these excitations cannot be detected by double braidings at all. They must be local excitations, and should be viewed as symmetry charges of a certain higher symmetry. Therefore, we regard a unitary symmetric fusion $n$-category as a physical higher symmetry.

\begin{expl} \label{rem:n-SFC}
There are more symmetric fusion $n$-categories than $n\Rep(G),n\Rep(G,z)$. We give some examples below. 
\bnu
\item For a finite abelian group $H$, $n\hilb_H$ has an obvious structure of a symmetric fusion $n$-category. Note that $2\hilb_H \nsimeq 2\Rep(G)$ as $2$-categories if $H$ is non-trivial because $2\Rep(G)$ is connected in the sense that $\hom_{2\Rep(G)}(i,j)\neq 0$ for any pair of simple objects $i,j\in2\Rep(G)$, but $2\hilb_H$ is completely disconnected. 
\item Let $\CG^{(n)}$ be a finite $n$-group. Then $\Rep(\CG^{(n)}) \coloneqq \fun(\CG^{(n)},n\hilb)$ has a canonical symmetric monoidal $n$-category structure inherited from that of $n\hilb$. We believe that it is also unitary fusion. In general, $\Rep(\CG^{(n)})$ has more symmetric fusion $n$-category structures even for $n=2$ \cite{tian}. This means that there are more symmetric fusion $n$-categories than those obtained from $n$-groups. 
\enu
\end{expl}

\begin{defn}
For $n\geq 1$, a unitary braided fusion $n$-category over $\CE$ is a pair $(\CC,\eta_\CC)$, where $\CC$ is a unitary braided fusion $n$-category and $\eta_\CC: \CE \hookrightarrow \FZ_2(\CC)$ a braided embedding, where $\FZ_2(\CC):=\Omega\FZ_1(\Sigma\CC)$ is the $E_2$-center (see \cite[Sec.\,IV.B]{jf20} and Remark\,\ref{rem:E2-center}). It is called a unitary modular $n$-category over $\CE$ if $\eta_\CC$ is an equivalence. If $\FZ_2(\CC)\simeq n\hilb$, it is called a unitary modular $n$-category. A pair $(\CM,\iota_\CM)$ is called a minimal modular extension of $(\CC,\eta_\CC)$ if $\CM$ is a unitary modular $n$-category and $\iota_\CM: \CC \hookrightarrow \CM$ is a braided embedding such that the canonical functor $\FZ_2(\CC) \to \FZ_2(\iota_\CM)$, where $\FZ_2(\iota_\CM)$ is the $E_2$-centralizer of $\iota_\CM$ (see Remark\,\ref{rem:E2-center}), is a braided equivalence. An equivalence $\phi: (\CM,\iota_\CM) \to (\CN,\iota_\CN)$ is a braided equivalence $\phi: \CM\to\CN$ such that $\phi\circ\iota_\CM\simeq \iota_\CN$. 
\end{defn}

\begin{rem}
Perhaps a better definition of a unitary modular $n$-category is a unitary braided fusion $n$-category $\CM$ such that the canonical functor $\CM \boxtimes \overline{\CM} \to \FZ_1(\CM)$ is a braided equivalence. Here $\overline{\CM}$ is the time reversal of $\CM$ (defined by flipping all highest morphisms), because this definition has a clear physical meaning. Similarly, a unitary modular $n$-category over $\CE$ is a unitary braided fusion $n$-category $\CC$  equipped with a braided equivalence from $\CC\boxtimes_\CE\overline{\CC}$ to the $E_2$-centralizer of $\CE\hookrightarrow\FZ_1(\CC)$ (see Remark\,\ref{rem:E2-center}). Both tensor product $\boxtimes$ and $\boxtimes_\CE$ have clear physical meanings and can be defined in certain Karoubi complete world \cite{jf20}. We believe that the mathematical foundation of multi-fusion $n$-categories provided in \cite{jf20} should lead to a proof of the equivalence of these two definitions. 
\end{rem}

\begin{pthm} \label{pthm:main-1}
For $n\geq 1$, we propose the following classification: 
\bnu
\item An anomaly-free $\nao$d (spatial dimension) SET order with a higher symmetry $\CE$ is described and classified by the equivalence classes (defined similarly as Definition~\ref{equiquad}) of a quadruple $(\CC,\eta_\CC,\CM,\iota_\CM)$, where $(\CC,\eta_\CC)$ is a unitary modular $n$-category over $\CE$, and $(\CM,\iota_\CM)$ is a minimal modular extension of $(\CC,\eta_\CC)$. 

\item When $\CC=\CE$, the set of equivalence classes of the pairs $(\CM,\iota_\CM)$, denoted by $\mext(\CE)$, classify all $\nao$d (spatial dimension) SPT$_{/\CE}$ orders with the symmetry $\CE$. The trivial $\nao$d SPT$_{/\CE}$ order is described by $(\FZ_1(\CE),\iota_0)$, where $\iota_0 : \CE \to \FZ_1(\CE)$ is the canonical embedding. When $\CE=n\Rep(G)$, the pair $(\FZ_1(n\hilb_G^\omega),\iota_\omega)$ for $\omega\in H^{n+2}(G,U(1))$ and a braided embedding $\iota_\omega: \CE \hookrightarrow \FZ_1(n\hilb_G^\omega)$, are examples (not all) of $\nao$d SPT$_{/\CE}$ orders (see Remark\,\ref{rem:wen14}), the pair $(\FZ_1(n\hilb_G^0)=\FZ_1(n\Rep(G)),\iota_0)$ describes the trivial $\nao$d SPT order. 
\enu
\end{pthm}

\begin{rem}
After gauging the symmetry $n\Rep(G)$, we obtain a family of ($\nao$)d topological orders $\FZ_1(n\hilb_G^\omega)$ also known as Dijkgraaf-Witten theories. It was conjectured in \cite{ktz} that there is an equivalence of $n$-categories: 
\be \label{eq:Z1-nhilb-w}
\FZ_1(n\hilb_G^\omega) \simeq \bigoplus_{[h]\in \mathrm{Cl}} \,\, n\Rep(C_G(h), \tau_h(\omega)), 
\ee
where $\mathrm{Cl}$ denotes the set of conjugacy classes of $G$, and $C_G(h)$ is the centralizer of $h\in G$, and $\tau_h: C^{n+2}(G,\Cb^\times) \to C^{n+1}(C_G(h),\Cb^\times)$ is the transgression map (see  \cite{will08}). In this context, $\iota_\omega$ is the embedding of $n\Rep(G)$ onto the $[1]$-component in $\FZ_1(n\hilb_G^\omega)$. The $n=1$ case of the conjecture (\ref{eq:Z1-nhilb-w}) was proved in \cite{will08}, and $n=2$ case was proved in \cite{ktz}. 
\end{rem}

\begin{expl}
For $\omega\in H^{4}(G,U(1))$, the braided fusion 2-category $\FZ_1(2\hilb_G^\omega)$ was explicitly computed in \cite{ktz}. We illustrate 4 simple objects in $\FZ_1(2\hilb_{\Zb_2}^0)$ by the following quiver: 
\begin{align} \label{quiver0}
\xymatrix{
\one \ar@(ul,ur)[]^{\Rep(\Zb_2)}  \ar@/^/[rr]^{1\hilb} & & T \ar@(ul,ur)[]^{1\hilb_{\Zb_2}} \ar@/^/[ll]^{1\hilb}
& & \one_s \ar@(ul,ur)[]^{\Rep(\Zb_2)}  \ar@/^/[rr]^{1\hilb} & & T_s \ar@(ul,ur)[]^{1\hilb_{\Zb_2}} \ar@/^/[ll]^{1\hilb}
},
\end{align}
where $\one$ is the tensor unit. The fusion rules are given by:
\begin{align*}
\one_s \boxtimes \one_s \simeq \one, \quad \one_s \boxtimes T \simeq T \boxtimes \one_s \simeq T_s, \quad
T \boxtimes T \simeq T_s \boxtimes T_s \simeq T \oplus T, \quad T \boxtimes T_s \simeq T_s \boxtimes T \simeq T_s \oplus T_s.
\end{align*}
See \cite[Example~3.8]{ktz} for the braiding structure. 
\end{expl}

\begin{rem} \label{rem:E2-center}
By \cite[Section\,5.3]{lurie}, the centralizer of a braided $n$-functor $f:\CA \rightarrow \CB$ is defined by the braided $n$-category $\FZ_2(f)$ (together with a braided $n$-functor $m$) that is universal among all commutative triangles:
$$
\xymatrix@R=1em{
& \FZ_2(f) \boxtimes \CA \ar[dr]^m & \\
\CA \ar[rr]^{f} \ar[ur]^{\one \boxtimes \id_\CA} &  & \CB\, .
}
$$
The $E_2$-center of $\CA$ is $\FZ_2(\id_\CA)$ which is automatically equipped with an $E_3$-monoidal structure. Importantly, $\FZ_2(\iota_\CM)$ is not a subcategory of $\CM$ because, even for $n=2$, an object in $\FZ_2(\iota_\CM)$ is an object in $\CM$ together with a family of 2-isomorphisms between double braidings and identity morphisms (see for example \cite[Def.\,3.10]{ktz}). Although we have nearly no concrete example of Theorem$^{\mathrm{ph}}$\,\ref{pthm:main-1} beyond the Dijkgraaf-Witten theories, the gauging-the-symmetry description of the trivial $\nao$d SPT$_{/\CE}$, i.e. the pair $(\FZ_1(\CE), \iota_0)$,  is enough to guarantee our second approach in Section\,\ref{sec:bb-relation} to work. We need the gauging-the-symmetry descriptions of non-trivial $\nao$d SPT$_{/\CE}$ orders only for SET orders with 't Hooft anomalies (see Def.\,\ref{defn:hooft-anomaly}). 
\end{rem}



\section{SPT/SET orders via boundary-bulk relation} \label{sec:bb-relation}
As we mentioned in the introduction, the idea of gauging the symmetry is not an intrinsic approach, and perhaps is rather strange from a mathematical point of view because an anomaly-free SPT/SET is well-defined before we gauging the symmetry. This means that some data intrinsically associated to the category of topological excitations is missing. It turns out that the missing data is living in the one-dimensional-higher bulk of the anomaly-free SPT/SET. 
\begin{pthm} \label{pthm:trivial-bulk}
The unique bulk of an anomaly-free $n$d SET order is the trivial $\nao$d SPT order. 
\end{pthm}
By the boundary-bulk relation, the categorical description of the $\nao$d bulk is given by the center of the category of topological excitations in the $n$d SPT/SET order, and should be identified with the non-trivial categorical description of the trivial $\nao$d SPT order obtained by gauging the symmetry in Section\,\ref{sec:mext}. This identification is not unique in general, and is precisely the missing data we are looking for. 

\medskip
In this section, we use this idea of boundary-bulk relation to obtain a classification of $n$d anomaly-free SPT/SET orders for $n\geq 0$. In Section\,\ref{sec:0d-SPT}, we study the 0d case, which is quite different from other cases. In Section\,\ref{sec:1d-SET-2}, we study the 1d case and show that the new classification results are compatible with the results in Section\,\ref{sec:1d-SET-1}. In Section\,\ref{sec:cc}, we introduce the physical notion of a condensation completion, which plays a crucial role in all $n$d cases for $n\geq 2$, and prove rigorously that it coincides with the mathematical notion of an ``idempotent completion'' introduced in \cite{dr} in the 2d case. In Section\,\ref{sec:2d-SET-2} and \ref{sec:nd-SET-2}, we discuss the 2d and higher dimensional cases, respectively.

\subsection{0d SPT/SET orders} \label{sec:0d-SPT}

If a 0d SET order is anomaly-free, it can be realized by a 0d lattice model. In this case, it makes no sense to talk about long range entanglement. All states are automatically short range entangled. Therefore, all anomaly-free 0d SET orders are SPT orders.

\medskip
Let $G$ be a finite onsite symmetry. The space of ground states of a gapped liquid state necessarily carries a symmetry charge, i.e. a representation of $G$. 
Let $\CE=\Rep(G)$ for a bosonic system or $\CE=\Rep(G,z)$ for a fermionic system. We argue in two different ways that the categorical description of a 0d SPT$_{/\CE}$ order is given by $\CE$, which is regarded as a category by forgetting its braiding and monoidal structure. As categories, we have $\CE=\Rep(G)=\Rep(G,z)$. 
\bnu
\item On the 0+1D world line, one can change the space of ground states from one representation of $G$ to another. This change can be achieved by inserting a 0D topological defect on the world line. As a consequence, all objects in $\Rep(G)$ can appear. Therefore, the categorical description of a 0d SPT$_{/\CE}$ order is $\CE=\Rep(G)$. 
\item According the boundary-bulk relation, the bulk of an anomaly-free 0d topological order must be the trivial 1d SPT order. By Theorem$^{\mathrm{ph}}$\,\ref{pthm:1d-spt-1}, the trivial 1d SPT order is described by the pair $(\fun_\hilb(\CE,\CE),\iota_0)$. Note that $\fun_\hilb(\CE,\CE)$ is precisely the $E_0$-center of $\CE$, i.e. 
\[
\FZ_0(\CE)=\fun_\hilb(\CE,\CE).
\] 
By the boundary-bulk relation \cite{kwz1,kwz2}, we see immediately that the categorical description of an anomaly-free 0d SPT order must be $\CE$. 
\enu

Using the idea of boundary-bulk relation, we obtain the following classification result. 
\begin{pthm}
A 0d SPT$_{/\CE}$ order is uniquely characterized by a pair $(\CE, \phi)$, where $\CE=\Rep(G)$ and $\phi: \FZ_0(\CE) \to \FZ_0(\CE)$ is a monoidal equivalence rendering the following diagram commutative. 
\be \label{diag:E-Z0-phi}
\raisebox{2em}{\xymatrix@R=1.5em{
& \CE \ar@{^(->}[dl]_{\iota_0} \ar@{^(->}[dr]^{\iota_0} & \\
\FZ_0(\CE) \ar[rr]_\simeq^\phi & & \FZ_0(\CE)\, .
}}
\ee
We denote the group of equivalence classes of such monoidal autoequivalences of $\FZ_0(\CE)$ by $\Aut^\otimes(\FZ_0(\CE),\iota_0)$. The trivial 0d SPT order is described by $(\CE,\id_{\FZ_0(\CE)})$. The group of 0d SPT orders is isomorphic to $\Aut^\otimes(\FZ_0(\CE),\iota_0)$. 
\end{pthm}

\begin{thm}
We have a canonical group isomorphism $\Aut^\otimes(\FZ_0(\Rep(G)),\iota_0) \simeq H^1(G,U(1))$. 
\end{thm}
\pf
Given $\phi\in\Aut^\otimes(\FZ_0(\CE),\iota_0)$, it endows $\CE=\Rep(G)$ with another left $\FZ_0(\CE)$-module structure $\odot:\FZ_0(\CE)\times \CE \to \CE$ defined by $a\odot x \coloneqq \phi(a)(x)$ for $a\in\FZ_0(\CE)$ and $x\in\CE$. We denote this module structure by ${}_\phi\CE$. Since $\CE$ is the unique left $\FZ_0(\CE)$-module up to equivalences, there is a module equivalence $f: \CE \to {}_\phi\CE$. We have $\phi(a)(f(x)) \simeq f(a(x))$. This implies that $\phi \cong f\circ a \circ f^{-1}$ (not canonically). For $e\in\CE$, we have $\iota_0(e)=e\otimes -$. The condition (\ref{diag:E-Z0-phi}) implies that 
\[
e\otimes f(x) \simeq \iota_0(e)(f(x)) \simeq \phi(\iota_0(e))(f(x)) \simeq f(e\otimes x), 
\]
which further implies that $f\in\fun_\CE(\CE,\CE)$. Using $\fun_\CE(\CE,\CE)\simeq \CE$, we can identify $f$ with $f(\one_\CE)\in\CE$. Since $f$ is an equivalence, $f(\one_\CE)$ is invertible, and is precisely a 1-dimensional representation of $G$, or equivalently, an element in $H^1(G,U(1))$. The composition of $\phi$ is compatible with that of $f$, which is further compatible with the tensor product of 1-dimensional representations of $G$ and with the multiplication in $H^1(G,U(1))$. 
\epf
This result is expected in physics.  The 0d invertible phases with symmetry $G$
are classified by 1-dimensional representations of $G$, which happen to be
given by  $H^1(G,U(1))$.

\begin{rem} \label{rem:theo-0d-spt}
There is an interesting discussion of 0d SPT orders, including the time reversal symmetry in \cite[Section\ V.A]{jf20}. 
\end{rem}

\subsection{1d SPT/SET orders revisit} \label{sec:1d-SET-2}
In this subsection, we give a new classification of all 1d anomaly-free SET$_{/\CE}$ orders based on the idea of boundary-bulk relation. 

\medskip
Let $\CE$ be a symmetric fusion 1-category of symmetry charges. We denote the category of topological excitations in a 1d anomaly-free SET$_{/\CE}$ order by $\CA$, which is a unitary fusion 1-category over $\CE$, i.e. $\CA$ is a unitary fusion category equipped with a braided embedding $\iota : \CE \hookrightarrow \FZ_1(\CA)$ such that the composed functor $\eta_\CA : \CE \stackrel{\iota}{\hookrightarrow} \FZ_1(\CA) \to \CA$ is fully faithful. According to the results in Section\,\ref{sec:2d-SET-1}, the trivial 2d SPT order is described by the pair $(\FZ_1(\CE),\iota_0)$. By Theorem$^{\mathrm{ph}}$\,\ref{pthm:trivial-bulk}, $\CA$ must be equipped with a braided equivalence $\phi: \FZ_1(\CE) \to \FZ_1(\CA)$. Moreover, $\phi$ must preserve the symmetry charges in $\CE$, i.e. $\phi$ is a braided equivalence rendering the following diagram commutative (up to natural isomorphisms): 
\be \label{diag:EA}
\raisebox{2em}{\xymatrix@R=1em{
& \CE \ar@{^(->}[dl]_{\iota_0} \ar@{^(->}[dr]^{\iota} & \\
\FZ_1(\CE) \ar[rr]_\simeq^\phi & & \FZ_1(\CA)\, .
}}
\ee
The main result in \cite{eno08} says that $\FZ_1(\CE)\simeq\FZ_1(\CA)$ if and only if $\CE$ and $\CA$ are Morita equivalent. This further implies that $\eta_\CA : \CE \to \CA$ is a monoidal equivalence, and thus $\CA$ can be identified with $\CE$ via $\eta_\CA$. We prove again the fact that all anomaly-free 1d SET orders are SPT orders.

\medskip
Now we show that the pair $(\CE,\phi)$ fully characterizes a 1d SPT order. No further data is needed. We achieve this by constructing a canonical isomorphism (recall Theorem$^{\mathrm{ph}}$\,\ref{pthm:1d-spt-1}):
\be \label{eq:pic-aut_0}
\mathrm{Pic}(\CE) \simeq \mathrm{Aut}^{br}(\FZ_1(\CE),\iota_0). 
\ee
This fact was first proved in \cite{dn}. Our proof has a clear physical meaning. 


In Figure\,\ref{fig:pic-aut} (a), we depict a physical configuration. Two 1d SPT orders are depicted as two 1d boundaries in Figure\,\ref{fig:pic-aut} (a). The excitations in these two 1d SPT orders are both given by $\CE$. Their 2d bulks are both the trivial 2d SPT order described by the pair $(\FZ_1(\CE),\iota_0)$. The invertible $\CE$-$\CE$-bimodule $\CX$ clearly describes a 0d domain wall between two 1d SPT orders. A 0d domain wall is defined by an invertible $\CE$-$\CE$-bimodule $\CX$. It uniquely determines a relative 1d bulk, i.e. a 1d domain wall between two trivial 2d SPT orders, defined by the unitary fusion category $\FZ_1^{(1)}(\CX) \coloneqq \fun_{\CE|\CE}(\CX,\CX)$ \cite{kz}. Since $\CX$ is invertible, $\FZ_1^{(1)}(\CX)$ is also an invertible domain wall \cite{kz}. Hence, this 1d wall $\FZ_1^{(1)}(\CX)$ must be a 1d SPT order. 

\begin{figure}
$$
\raisebox{-30pt}{
  \begin{picture}(120,95)
   \put(0,15){\scalebox{1}{\includegraphics{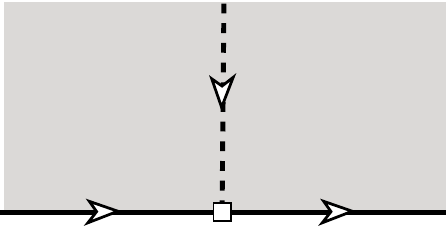}}}
   \put(0,15){
     \setlength{\unitlength}{.75pt}\put(0,0){
     \put(130,48)  {\scriptsize $(\FZ_1(\CE),\iota_0)$}
     \put(10,48) {\scriptsize $(\FZ_1(\CE),\iota_0)$}
     \put(35,-8) {\scriptsize $\CE$}
     \put(125,-8){\scriptsize $\CE$}
     \put(81,-7) {\scriptsize $\CX$}
     \put(73,93) {\scriptsize $\FZ_1^{(1)}(\CX) \coloneqq \fun_{\CE|\CE}(\CX,\CX)$}
         
     }\setlength{\unitlength}{1pt}}
  \end{picture}} 
\quad\quad\quad\quad\quad\quad\quad\quad
\raisebox{-30pt}{
  \begin{picture}(120,95)
   \put(0,15){\scalebox{1}{\includegraphics{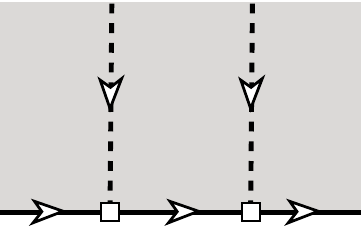}}}
   \put(0,15){
     \setlength{\unitlength}{.75pt}\put(0,0){
     \put(38,-8) {\scriptsize $\CX_1$}
     \put(93,-9) {\scriptsize $\CX_2$}
     \put(15,-8) {\scriptsize $\CE$}
     \put(66,-8) {\scriptsize $\CE$}
     \put(116,-8) {\scriptsize $\CE$}
     \put(39,93) {\scriptsize $\CY_{\phi_1}$}
     \put(93,93) {\scriptsize $\CY_{\phi_2}$}
         
     }\setlength{\unitlength}{1pt}}
  \end{picture}} 
$$
$$
(a) \hspace{6cm} (b)
$$
\caption{Picture (a) illustrates the relation among $\CE,\CA,\CX,\CM$ in a physical way and provides a proof of the canonical isomorphism in (\ref{eq:pic-aut}); Picture (b) illustrates the compatibility between the multiplications in $\mathrm{Pic}(\CE)$ and $\Aut^{br}(\FZ_1(\CE),\iota_0)$, where $\CY_{\phi_i}$ denotes the invertible 1d domain wall associated to $\phi_i$ for $i=1,2$ and is itself a 1d SPT order. 
}
\label{fig:pic-aut}
\end{figure}

\medskip
According to \cite[Theorem\ 3.3.7]{kz}, the assignment
\be \label{eq:Z-functor}
\CE \mapsto \FZ_1(\CE), \quad\quad \CX \mapsto \FZ_1^{(1)}(\CX) \coloneqq \fun_{\CE|\CE}(\CX,\CX)
\ee
defines a fully faithful functor. As a consequence, there is a one-to-one correspondence between the equivalence classes of invertible $\CE$-$\CE$-bimodules and those of invertible domain walls between two 2d topological orders defined by $\FZ_1(\CE)$. It is also known that there is a natural group isomorphism from the set of equivalence classes of braided autoequivalences of $\FZ_1(\CE)$ to the set of equivalence classes of invertible domain walls between two $\FZ_1(\CE)$'s \cite[Example\ 2.6.7,Corollary\ 3.3.10]{kz}. We denote it by $\phi \mapsto \CY_\phi$. This leads to the well-known group isomorphism (first proved in \cite{eno09}):
\be \label{eq:brpic-aut}
\FZ_1^{(1)}|_{\mathrm{BrPic}(\CE)}: \,\, \mathrm{BrPic}(\CE) \xrightarrow{\simeq} \mathrm{Aut}^{br}(\FZ_1(\CE)). 
\ee
Moreover, we can see directly from Figure\,\ref{fig:pic-aut} that the condition for $\CX\in\Pic(\CE)$ is equivalent to the condition for $\phi\in\mathrm{Aut}^{br}(\FZ_1(\CE),\iota_0)$, where $\phi$ is determined by the 1d SPT order $\CY_\phi\simeq\FZ_1^{(1)}(\CX)$. Therefore, 
we obtain a group isomorphism 
\be \label{eq:pic-aut}
\FZ_1^{(1)}|_{\Pic(\CE)}: \,\, \mathrm{Pic}(\CE) \xrightarrow{\simeq} \mathrm{Aut}^{br}(\FZ_1(\CE),\iota_0).
\ee

\medskip
Notice that the stacking of two 1d SPT orders amounts to stacking two layers of
Figure\,\ref{fig:pic-aut}. Recall that the stacking of the 2d bulks is defined
by (\ref{eq:def-star}). All 1d and 0d defects should stack compatibly. We know
two 1d SPT orders stack according to $\CE\boxtimes_\CE\CE \simeq \CE$. This is
compatible with (\ref{eq:def-star}). As a consequence, $\CX_1$ and $\CX_2$ must
stack according to $\CX_1\boxtimes_\CE\CX_2$, which is compatible with the composition in $\mathrm{Aut}^{br}(\FZ_1(\CE),\iota_0)$ as illustrated in Figure\,\ref{fig:pic-aut} (b), which shows nothing but the functoriality of (\ref{eq:Z-functor}). Therefore, we obtain a complete classification of 1d SPT/SET orders. 

\begin{pthm}\label{pthm:1d-spt-2}
All anomaly-free 1d SET$_{/\CE}$ orders are 1d SPT$_{/\CE}$ orders. A 1d SPT$_{/\CE}$ order can be
characterized either by $\phi\in\mathrm{Aut}^{br}(\FZ_1(\CE),\iota_0)$, or by
$\CY_\phi$, or by $(\FZ_1^{(1)})^{-1}(\CY_\phi)\in\Pic(\CE)$. Moreover, the
group of 1d SPT$_{/\CE}$ orders is isomorphic to both $\mathrm{Aut}^{br}(\FZ_1(\CE),\iota_0)$ and $\Pic(\CE)$. 
\end{pthm}

Although $\CA$ and $\CM$ are not explicitly in Figure\,\ref{fig:pic-aut} (a), they can be recovered by fusing 1d phases along 2d phases.  
\bnu
\item $\CA$ can be recovered by fusing $\FZ_1^{(1)}(\CX)$ with the right 1d boundary SPT order $\CE$ along the 2d bulk $(\FZ_1(\CE),\iota_0)$, i.e. 
\be \label{eq:recover-A}
\fun_{\CE|\CE}(\CX,\CX) \boxtimes_{\FZ_1(\CE)} \CE \simeq \fun_\CE(\CX,\CX) \simeq \CA, 
\ee
where the first monoidal equivalence is due to \cite[Theorem\,3.1.7]{kz}.

\item $\CM$ can be recovered by closing the fan to give an anomaly-free 1d phase defined by 
\be \label{eq:recover-M}
\CE\boxtimes_{\FZ_1(\CE)} \FZ_1^{(1)}(\CX) \boxtimes_{\FZ_1(\CE)} \CE \simeq 
\CE\boxtimes_{\FZ_1(\CE)} \fun_\CE(\CX,\CX) \simeq \fun_\hilb(\CX,\CX)=\CM. 
\ee
where the second monoidal equivalence is again due to \cite[Theorem\,3.1.7]{kz}.
\enu

\begin{rem}\label{rem:spt-commutative}
It is easy to see why the group $\Pic(\CE)$ or $\Aut^{br}(\FZ_1(\CE),\iota_0)$ is abelian. This follows from the following  natural equivalences:
\[
\CY_{\phi_1 \circ \phi_2} \simeq \CY_{\phi_1} \boxtimes_{(\FZ_1(\CE),\iota_0)} \CY_{\phi_2} \simeq 
\CY_{\phi_2} \boxtimes_{(\overline{\FZ_1(\CE)},\overline{\iota_0})} \CY_{\phi_1 \simeq } \simeq 
\CY_{\phi_2} \boxtimes_{(\FZ_1(\CE),\iota_0)} \CY_{\phi_1} \simeq \CY_{\phi_2\circ \phi_1}, 
\]
where the second $\simeq$ is obtained by doing a left-right mirror reflection of Figure\,\ref{fig:pic-aut} (b), and the third $\simeq$
is due to the fact that $(\overline{\FZ_1(\CE)},\overline{\iota_0}) \simeq (\FZ_1(\CE),\iota_0)^{-1} \simeq (\FZ_1(\CE),\iota_0)$ \cite{LW160205936}, or equivalently, the trivial SPT order preserves the time-reversal symmetry. 
\end{rem}

Actually, the physical stacking of two trivial 2d SPT orders induces an independent mathematical definition of a new multiplication on $\mathrm{Aut}^{br}(\FZ_1(\CE),\iota_0)$. Physically, this multiplication must coincide with the composition of functors in 
$\mathrm{Aut}^{br}(\FZ_1(\CE),\iota_0)$. This leads to a non-trivial mathematical result, which should be proved independently and rigorously. We spell out this result explicitly in Theorem$^{\mathrm{ph}}$\,\ref{pthm:star=composition}. 

\medskip
There is a well-defined map 
\begin{align*}
\Aut^{br}(\CM,\iota_\CM) \times \Aut^{br}(\CN,\iota_\CN) &\xrightarrow{\star} \Aut^{br}((\CM \boxtimes \CN)_{L_\CE}^0,\iota_\CM \star \iota_\CN) \\
(\phi_1, \phi_2) &\mapsto \phi_1 \star \phi_2 \coloneqq (x \mapsto
(\phi_1\boxtimes\phi_2)(x)).
\end{align*}
Indeed, it is clear that $\phi_1\star\phi_2 \in\Aut^{br}((\CM \boxtimes \CN)_{L_\CE}^0)$, and we have, for $x\in (\CE\boxtimes\CE)_{L_\CE}^0$, 
\begin{align*}
(\phi_1\star \phi_2)\circ (\iota_\CM \star \iota_\CN)(x) & \coloneqq (\phi_1\boxtimes\phi_2)((\iota_\CM \boxtimes \iota_\CN)(x)) \\
&= \left( (\phi_1 \circ \iota_\CM) \boxtimes (\phi_2 \circ \iota_\CN)\right)(x) \simeq (\iota_\CM \boxtimes \iota_\CN)(x).
\end{align*} 
Since the trivial 2d SPT order $(\FZ_1(\CE), \iota_0)$ gives the identity element under $\star$, for a minimal modular extension $(\CM,\iota_\CM)$, there is a canonical braided equivalence $g: \CM \to (\CM\boxtimes\FZ_1(\CE))_{L_\CE}^0$ explicitly constructed in \cite[Proof of Lemma 4.18]{LW160205936}. Using $g$, we obtain a map 
\begin{align*}
\Aut^{br}(\CM,\iota_\CM) \times \Aut^{br}(\FZ_1(\CE),\iota_0) &\xrightarrow{\star^g} \Aut^{br}(\CM,\iota_\CM) \\
(\phi_1, \phi_2) &\mapsto \phi_1\star^g \phi_2 \coloneqq g^{-1}\circ (\phi_1\star\phi_2) \circ g
\end{align*} 
\begin{pthm} \label{pthm:star=composition}
When $(\CM,\iota_\CM)=(\FZ_1(\CE),\iota_0)$, we have $\phi_1 \star^g \phi_2 \simeq \phi_1 \circ \phi_2$. 
\end{pthm}


\medskip
In the rest of this subsection, we discuss 1d SET$_{/\CE}$ orders with anomalies. In general, the anomaly of a SET order is a mixture of that associated to symmetries and the gravitational anomaly. It is difficult to distinguish them by a clean definition except in some special cases, such as an SET order obtained by stacking an SET order with only gravitational anomaly with another SET order with only anomaly associated to symmetries. However, it is possible to define them separately in certain limits. 

\begin{defn} \label{defn:gravitational-anomaly}
The gravitational anomaly of an SET order $\CX$ is defined by that of the bosonic or fermionic topological order obtained by fully breaking the symmetry in $\CX$. 
\end{defn}

\begin{defn} \label{defn:hooft-anomaly}
An SET order without gravitational anomaly is said to have a non-trivial (resp. trivial) {\it symmetry anomaly} if its bulk is a 1-dimensional-higher non-trivial (resp. trivial) SPT order (without any intrinsic topological order). If its bulk SPT order is a non-trivial twisted gauge theory, this symmetry anomaly is called {\it 't Hooft anomaly}. 
\end{defn}

\begin{pthm} \label{pthm:1d-thooft}
\label{SET1d}
If a 1d SET order (modulo invertible topological orders) does not have any
gravitational anomaly, then it does not have any non-trivial symmetry anomaly.
\end{pthm}
\begin{proof}
The category of topological excitations of a 1d SET$_{/\CE}$ order with a
symmetry anomaly still form a unitary fusion 1-category $\CA$ over $\CE$, i.e. a fusion category $\CA$ equipped with a braided embedding $\iota : \CE \hookrightarrow \FZ_1(\CA)$ such that the composed functor $\eta_\CA : \CE \stackrel{\iota}{\hookrightarrow} \FZ_1(\CA) \to \CA$ is fully faithful.
By Definition\,\ref{defn:hooft-anomaly}, we obtain a commutative diagram: 
\be \label{diag:EMA}
\raisebox{2em}{\xymatrix@R=1em{
& \CE \ar@{^(->}[dl]_{\iota_\CM} \ar@{^(->}[dr]^{\iota} & \\
\CM \ar[rr]_\simeq^\phi & & \FZ_1(\CA)\, ,
}}
\ee
where $(\CM,\iota_\CM)$ is a minimal modular extension of $\CE$ describing the
symmetry anomaly and $\phi$ is a braided equivalence. Since the quantum dimension of $\CM$ is the same as that of $\FZ(\CE)$ by the definition of minimal modular extensions, we obtain $\fpdim(\CA) = \fpdim(\CE)$, which further implies that $\eta_\CA : \CE \hookrightarrow \CA$ is a monoidal equivalence. Therefore, $(\CM,\iota_\CM) \simeq (\FZ(\CE),\iota_0)$ is the trivial SPT. 
\end{proof}

\begin{rem}
The proof of Theorem$^{\mathrm{ph}}$\,\ref{pthm:1d-thooft} can be viewed as a mathematical proof of an earlier result which says that the boundary of a 2d non-trivial SPT order must be gapless or symmetry breaking \cite{CLW1141}. 
\end{rem}

\begin{rem}
It turns out that $n$d SET orders with the mixture of gravitational anomalies and symmetry
anomalies can be characterized in a similar manner. See Remark \ref{anoset} for
details.
\end{rem}

\subsection{Condensation completion} \label{sec:cc}

We would like to find a categorical description of 2d SPT/SET orders using the idea of boundary-bulk relation. For this purpose, we need to find a categorical description of an anomaly-free 2d SET order regarded as a 2d boundary of the trivial 3d SPT order.

\medskip
Could this yet-to-be-found categorical description be the same as we have seen? One of the important lessons we have learned in \cite{KW14,kwz1} is the following. 
\begin{quote}
The categorical description of a potentially anomalous $n$d topological order $\CP_n$ depends on its codimension with respect to an ($n$+$k$)d anomaly-free topological order $\CQ_{n+k}$, in which $\CP_n$ is realized as a $k$-codimensional gapped defect. 
\end{quote}
A 0-codimensional description is only possible when the topological order (e.g. $\CQ_{n+k}$) is anomaly-free. In this work, the 0-codimensional description of an anomaly-free topological order $\CQ_{n+k}$ is always chosen to be a unitary modular ($n$+$k$-1)-category, i.e. an $E_2$-algebra.  

\begin{rem}
The trivial 1-codimensional domain wall in the topological order $\CQ_{n+k}$ naturally inherits a 1-codimensional description as an $E_1$-algebra (by forgetting the braidings), its $k$-codimensional defects living on the trivial 1-codimensional domain wall naturally inherit a $k$-codimensional description as $E_{2-k}$-algebras (see \cite[Remark\,2.24]{kwz1} for the meaning of an $E_{-1}$-algebra, an $E_{-2}$-algebra, etc). 
\end{rem}

Using dimensional reduction, a potentially anomalous $n$d topological order $\CP_n$ can always be realized as a gapped boundary of an anomaly-free $\nao$d topological order \cite{kwz1}. Therefore, a 1-codimensional description of $\CP_n$ is always possible. In this work, we only care about 0-codimensional and 1-codimensional descriptions, which are often different. 

\begin{defn}[\cite{KW14,kwz1}]
For a potentially anomalous $n$d topological order $\CP_n$, a 1-codimensional description of $\CP_n$ is a unitary fusion $n$-category, i.e. an $E_1$-algebra, such that its monoidal center (or $E_1$-center) coincides with the 0-codimensional description of the anomaly-free $\nao$d bulk of $\CP_n$. 
\end{defn}

It is well-known that an anomaly-free 2d topological order (modulo invertible topological orders, which are $E_8$-states in this case) can be described by a UMTC $\CM$. Note that $\CM$ only describes particle-like excitations. This is possible not because there is no other topological excitations. Actually, in general, there are many gapped 1d domain walls and 0d walls between 1d walls in an anomaly-free 2d topological order. The reason we can ignore them is because they can all be obtained from particle-like excitations via condensations as shown in \cite{kong,KW14}. Therefore, in this case, we can regard particle-like excitations as more elementary, and view all the rest topological excitations as descendants of the elementary ones. In this sense, the UMTC $\CM$ gives a 0-codimensional description of this 2d topological order. 

However, if we want to regard the same anomaly-free 2d topological order $\CM$ as a 1-codimensional gapped boundary of the trivial 3d topological order and check the boundary-bulk relation \cite{kwz1,kwz2}, the UMTC $\CM$ is not enough. We need to find a 1-codimensional description of the same anomaly-free 2d topological order. One of the lessons we have learned in \cite{kz19a,kz19b} is that a mathematical description of 1-codimensional (gapped or gapless) boundary should include all possible topological defects and all condensation descendants. 
\bnu
\item For example, the 0-codimensional description of an anomaly-free 1+1D rational CFT can simply be a non-chiral CFT with modular invariant partition functions\footnote{Strictly speaking, this is not correct. The correct one is a non-chiral CFT with structure constants satisfying genus-0 factorization properties and with modular invariant 1-point correlation function on torus.}. On the other hand, its 1-codimensional description, when viewed as a gappable gapless boundary of the trivial 2d topological order, must includes all possible 0+1D domain walls and 0D walls between two 0+1D walls allowed by a given non-chiral symmetry \cite[Section\ 5.2]{kz19b}. As a consequence, the complete set of defects of all dimensions forms an enriched fusion category $\CXs$, whose monoidal center gives precisely the 0-codimensional description of the trivial 2d topological order, i.e. $\FZ_1(\CXs) \simeq 1\hilb$.
\item  Similarly, the mathematical description of a chiral gapless boundary of an anomaly-free 2d chiral topological order $\CM$, i.e. a 1-codimensional description, should not only include those chiral fields living on the entire 1+1D world sheet of the 1d boundary but also all 0+1D domain walls and 0D walls between 0+1D walls on the same 1+1D world sheet \cite{kz19a}. The complete set of defects of all dimensions forms again an enriched (multi-)fusion category $\CYs$ such that $\FZ_1(\CYs) \simeq \CM$. 
\enu
In our 2d case, all possible topological defects include not only particle-like topological excitations but also 1d gapped domain walls and 0d walls between 1d walls. Since all topological excitations beyond those in $\CM$ can all be obtained from those in $\CM$ via condensations, the process of including these condensation descendants can be called the condensation completion of $\CM$ (see also the discussion in \cite[Section\ XI.B.]{KW14}). We summarize and conclude as follows.  
\begin{pthm}
The 1-codimensional (categorical) description of an anomaly-free 2d topological order $\CM$ is given by the condensation completion of $\CM$. 
\end{pthm}

We should view this as a special case of a general principle.  
\begin{quote}
{\bf Condensation-Completion Principle}: All possible defects (including condensation descendants) should be included in the 1-codimensional (categorical) description of a potentially anomalous $n$d gapped liquid phases when it is viewed as a boundary of its 1-higher-dimensional bulk. Moreover, the 1-codimensional description of the boundary and the 0-codimensional description of the bulk satisfies the boundary-bulk relation (i.e. bulk is the center of the boundary). 
\end{quote}

\begin{rem}
We believe that the not-yet-constructed higher dimensional generalization of Levin-Wen models \cite{LW05} with gapped boundaries can provide concrete realizations of about Condensation-Completion Principle. The 2d cases were done in \cite{KK12}. A special family of cases in 3d were realized in \cite{WW11}. 
\end{rem}

The following Theorem$^{\mathrm{ph}}$ generalizes a similar result for non-chiral UMTC's obtained in \cite[Section\ XI.B, Remark\ 16]{KW14} to all UMTC's. 
\begin{pthm} \label{pthm:cc-M}
The condensation completion of $\CM$ is given by the delooping $\Sigma\CM$ of $\CM$, which is equivalent to the fusion 2-category $\RMod_\CM(2\hilb)$ of right $\CM$-modules in $2\hilb$. 
\end{pthm}
\pf
Since $\CM$ is a UMTC, we have $\CM \simeq \CM^\rev$ as fusion categories via its braidings. It is easy to see that $\RMod_{\CM}(2\hilb) \simeq \LMod_{\CM}(2\hilb)$.
 
When $\CM$ is viewed as a unitary fusion 1-category by forgetting its braidings, it defines the trivial 1d domain wall as illustrated in Figure\,\ref{fig:cond-completion} along with some other 1d, 0d domain walls. All the other gapped 1d domain walls in the same topological order can be described by unitary multi-fusion 1-categories $\fun_{\CM}(\CX,\CX)$ of $\CM$-module functors for a finite unitary right $\CM$-module $\CX$ (see \cite{dmno,kong,egno}). The 0d domain wall between two such 1d walls $\fun_{\CM}(\CX,\CX)$ and $\fun_{\CM}(\CY,\CY)$ is precisely given by $\fun_{\CM}(\CX,\CY)$ or $\fun_{\CM}(\CY,\CX) \simeq \fun_{\CM}(\CX,\CY)^\op$ depending on the orientation of 1d walls (see \cite{akz}).

\begin{figure}
$$
\raisebox{-30pt}{
  \begin{picture}(150,95)
   \put(-20,15){\scalebox{0.9}{\includegraphics{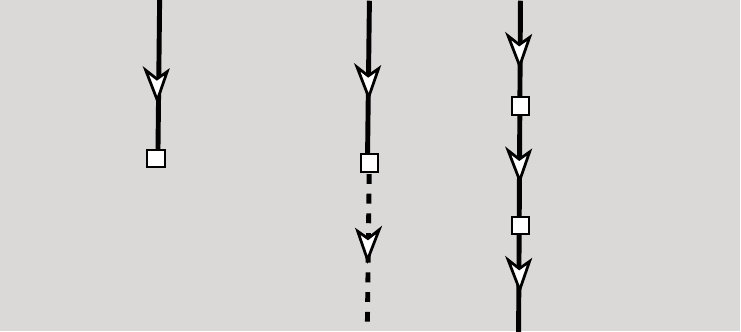}}}
   \put(-20,15){
     \setlength{\unitlength}{.75pt}\put(0,0){
     \put(245,3)  {\scriptsize $\CM$}
     \put(10,8) {\scriptsize $\CM$}
     \put(110,28) {\scriptsize $\CM$}
     \put(-15,48) {\scriptsize $\fun_{\CM}(\CM,\CX)\simeq\CX$}
     \put(116,56) {\scriptsize $\CY$}
     \put(186,97) {\scriptsize $\fun_{\CM}(\CX_3,\CX_3)$}
     \put(186,76) {\scriptsize $\fun_{\CM}(\CX_2,\CX_3)$}
     \put(186,55) {\scriptsize $\fun_{\CM}(\CX_2,\CX_2)$}
     \put(186,35) {\scriptsize $\fun_{\CM}(\CX_1,\CX_2)$}
     \put(186,15) {\scriptsize $\fun_{\CM}(\CX_1,\CX_1)$}
     \put(35,121) {\scriptsize $\fun_{\CM}(\CX,\CX)$}
     \put(110,121) {\scriptsize $\fun_{\CM}(\CY,\CY)$}
         
     }\setlength{\unitlength}{1pt}}
  \end{picture}} 
$$
\caption{\small This picture illustrates some 1d, 0d domain walls in an anomaly-free 2d topological order $\CM$. 
}
\label{fig:cond-completion}
\end{figure}

We illustrate 0-,1-morphisms in the 2-category $\Sigma\CM$ in the following diagram:
\begin{align} \label{quiver1}
\xymatrix{
\CX \ar@(ul,ur)[]^{\fun_{\CM}(\CX,\CX)} \ar@/^/[rr]^{\fun_{\CM}(\CX,\CY)} & & \CY \ar@(ul,ur)[]^{\fun_{\CM}(\CY,\CY)} \ar@/^/[ll]^{\fun_{\CM}(\CY,\CX)}
} .
\end{align}
\bnu
\item Each 0-morphism $\CX\in\Sigma\CM$ labels a 1d domain wall $\fun_{\CM}(\CX,\CX)$.  

\item 1-morphisms in $\Sigma\CM$ are precisely 0d walls or 0d topological excitations in 1d walls.  

\item 2-morphisms in $\Sigma\CM$ are instantons.

\item The composite $\circ$ of 1-morphisms in $\Sigma\CM$ is precisely the fusion of 0d walls along 1d walls (e.g. the vertical fusion $\boxtimes_{\fun_\CM(\CX_2,\CX_2)}$ in Figure\,\ref{fig:cond-completion}) as shown by the following commutative diagram: 
\be \label{diag:vertical-fusion}
\xymatrix{
\fun_\CM(\CX_2,\CX_3) \boxtimes \fun_\CM(\CX_1,\CX_2) \ar[r]^-\circ \ar[d]_{\boxtimes_{\fun_\CM(\CX_2,\CX_2)}} & \fun_\CM(\CX_1,\CX_3) \\
\fun_\CM(\CX_2,\CX_3) \boxtimes_{\fun_\CM(\CX_2,\CX_2)} \fun_\CM(\CX_1,\CX_2)\, \ar[ur]_\simeq & 
} .
\ee
\enu
We conclude that the 2-category $\Sigma\CM$ precisely encodes the information of all 0d,1d domain walls in $\CM$ and the vertical fusion among them. 

Notice that 0d,1d domain walls can also be fused horizontally in Figure\,\ref{fig:cond-completion}, e.g. 
\be \label{eq:horizontal-fusion}
\fun_\CM(\CX,\CX) \boxtimes_\CM \fun_\CM(\CY,\CY), \quad\quad \fun_\CM(\CX,\CX') \boxtimes_\CM \fun_\CM(\CY,\CY').
\ee
What structure in $\Sigma\CM$ encodes these horizontal fusions? It turns out that since $\CM$ is braided (i.e. an $E_2$-algebra), $\Sigma\CM$ has an additional monoidal structure (i.e. an $E_1$-algebra). Moreover, according to \cite{dr}, $\Sigma\CM$ is a fusion 2-category. We show that this monoidal structure encodes the information of the horizontal fusion of 1d,0d domain walls in $\CM$. 
\bnu
\item The monoidal structure $\otimes$ on $\Sigma\CM$ is defined by $\CX \boxtimes_\CM \CY$ on 0-morphisms. Notice that this is compatible with the horizontal fusion of two 0d walls in Figure\,\ref{fig:cond-completion}. 
\item The monoidal structure on higher morphisms is defined by the functor 
\begin{align} \label{eq:otimes-morphisms}
\fun_\CM(\CX,\CX') \boxtimes \fun_\CM(\CY,\CY') &\xrightarrow{\otimes} \fun_\CM(\CX\boxtimes_\CM\CY, \CX'\boxtimes_\CM\CY'),  \nn
f\boxtimes g &\mapsto  (f\boxtimes_\CM g: x\boxtimes_\CM y \mapsto f(x)\boxtimes_\CM g(y)).  
\end{align}
which is monoidal when $\CX=\CX'$ and $\CY=\CY'$. It coincides with the horizontal fusion (\ref{eq:horizontal-fusion}) as shown by the following commutative diagram: 
\be \label{diag:otimes-morphisms}
\xymatrix{
\fun_\CM(\CX,\CX') \boxtimes \fun_\CM(\CY,\CY') \ar[d]_{\boxtimes_\CM} \ar[r]^{\otimes} &  \fun_\CM(\CX\boxtimes_\CM\CY, \CX'\boxtimes_\CM\CY') \\
\fun_\CM(\CX,\CX') \boxtimes_\CM \fun_\CM(\CY,\CY'), \ar[ur]_\simeq & 
}
\ee
\enu
where the equivalence ``$\simeq$'' is monoidal if $\CX=\CX'$ and $\CY=\CY'$. The commutativity of the diagram and the  equivalence $\simeq$ follow from \cite[Theorem\ 3.3.6]{kz} (using the canonical faithful functor $\fun_\CM(\CX,\CX')\to\fun_{\CM|\CM}(\CX,\CX')$). 
\epf

\begin{rem} \label{rem:Theo-2}
The idea of condensation completion was first discussed in a physical context in \cite[Section\ XI]{KW14}, where condensation descendants are called condensed excitations, and the terminology of ``condensation completion'' was not used. Theorem$^{\mathrm{ph}}$\,\ref{pthm:cc-M} for non-chiral UMTC's was obtained in \cite[Remark\ 16]{KW14}, but was stated in a different but equivalent way according to \cite{kz,KYZ19}. 
Theorem$^{\mathrm{ph}}$\,\ref{pthm:cc-M} explains the physical meaning of the mathematical notion of ``idempotent completion'' introduced in mathematics by Douglas and Reutter in \cite{dr} with the motivation of making a 2-category semisimple. The physical necessity of the idempotent completion was further convinced in 3d $G$-gauge theory \cite{ktz}. This notion was further generalized to higher categories by Gaiotto and Johnson-Freyd \cite[Definition\ 1.3.1, 2.1.1]{gjf19} (see also \cite{jf20}) under the name of ``Karoubi completion'', which is briefly reviewed in Appendix. 
\end{rem}

\begin{rem} \label{rem:sigma-RMod}
There is a subtle difference between $\Sigma\CM$ and $\RMod_\CM(2\hilb)$ because there is no canonical functor from $\RMod_\CM(2\hilb)$ to $\Sigma\CM$. The more careful treatment of condensation completion requires us to select a distinguished object $x\in\CX$. This replaces $\CX$ by the pair $(\CX,x)$ as an object. In this way, we recover the Karoubi completion $\Sigma\CM$ in \cite{dr,gjf19}. 
\end{rem}

\begin{rem} \label{rem:premodular}
That the condensation completion is given by $\Sigma\CM$ remains to be true if $\CM$ is pre-modular. In this case, one can view $\CM$ as the mathematical description of particle-like topological excitations in an anomalous 2d topological order. Note that $\CM$ viewed as a unitary fusion 1-category describes again the trivial 1d domain wall in this anomalous 2d topological order. All other 1d domain walls are again given by $\fun_\CM(\CX,\CX)$. In particular, Figure\,\ref{fig:cond-completion} continues to make sense by regarding the picture as a 2d boundary of a hidden non-trivial 3d bulk.  
\end{rem}

We assume that the condensation completion for higher dimensional topological orders is given by the mathematical theory of Karoubi completion developed in \cite{gjf19}. Then we can use the same notation of the delooping to be that of condensation completion. Moreover, we assume that $\Sigma\CD$ is unitary for a unitary (braided) multi-fusion $n$-category.

\begin{expl}
We give a few examples. First, $\Sigma\hilb=2\hilb$, $\Sigma n\Rep(G) = (\nao)\Rep(G)$, $\Sigma n\Rep(G,z) = (\nao)\Rep(G,z)$ by definitions. 
In particular, $\Omega^{n-1}(n\Rep(G))=\Rep(G)$ and
$\Omega^{n-1}(n\Rep(G,z))=\Rep(G,z)$. This simply says that the physical meaning of $(\nmo)$-morphisms in $n\Rep(G)$ or $n\Rep(G,z)$ are symmetry charges, and all the rest morphisms in $n\Rep(G)$ or $n\Rep(G,z)$ are the condensation descendants of the symmetry charges. Each time we increase the dimension by one, all 1-higher-dimensional condensation descendants must be included in the condensation completion. 
\end{expl}

Since $\Sigma\CM$ contains all possible topological defects in the 2d boundary of the trivial 3d topological order, and by \cite{kwz1,kwz2}, we obtain the following result.
\begin{pthm} \label{pthm:center-repM=trivial}
$\FZ_1(\Sigma\CM) \simeq 2\hilb$ as braided 2-categories. 
\end{pthm}

\begin{rem}
We briefly sketch the idea of the proof of Theorem$^{\mathrm{ph}}$\,\ref{pthm:center-repM=trivial} here. Using the proof of Theorem$^{\mathrm{ph}}$\,\ref{thm:fun-PP} and Remark\,\ref{rem:fun-PP-proof-2}, one can show that an object in $\FZ_1(\Sigma\CM)$ is necessarily a direct sum of the identity functor $\id_{\Sigma\CM}$. Then the theorem follows from the fact that $\Omega \FZ_1(\Sigma\CM) \simeq 1\hilb$. We will provide details elsewhere. See Remark\,\ref{rem:theo-umtc}. During the second revision request by JHEP, Davydov and Nikshych posted a rigorous proof of this result in \cite[Thm.\ 4.10,Prop.\ 4.16]{DN20}. 
\end{rem}

If a unitary fusion 2-category $\CT$ has a trivial monoidal center, i.e. $\FZ_1(\CT)\simeq 2\hilb$, it means that $\CT$ describes an anomaly-free 2d topological order. Since 1-codimensional defects cannot be detected by braidings, they must be condensation descendants of particle-like excitations, which can be braided. Therefore, we obtain the following result. 
\begin{pthm} \label{pthm:T=sigmaM}
A unitary fusion 2-category $\CT$ has the trivial monoidal center if and only if there is a UMTC $\CM$ such that $\CT \simeq \Sigma\CM$ as fusion 2-categories. 
\end{pthm}

\begin{rem} \label{rem:theo-umtc}
The non-unitary version of Theorem$^{\mathrm{ph}}$\,\ref{pthm:center-repM=trivial} and \ref{pthm:T=sigmaM} is proved as the 2d case of a general result for all higher dimensions in \cite[Corollary\ IV.2]{jf20} (see also Remark\,\ref{rem:theo-n-umtc}). 
\end{rem}

If $\CM$ is non-chiral, then $(\CM,0)$ describes a non-chiral 2d topological order that admits gapped 1d boundaries. Each gapped 1d boundary can be described by a closed multi-fusion left $\CM$-module $\CP$ \cite[Definition\ 2.6.1]{kz}, i.e. a unitary multi-fusion 1-category $\CP$ equipped with a braided equivalence $\psi_\CP: \CM \xrightarrow{\simeq} \FZ_1(\CP)$. Any two such closed multi-fusion left $\CM$-modules $\CP$ and $\CQ$ are necessarily Morita equivalent \cite{eno08}.

\begin{figure}
$$
\raisebox{-30pt}{
  \begin{picture}(120,85)
   \put(0,15){\scalebox{1}{\includegraphics{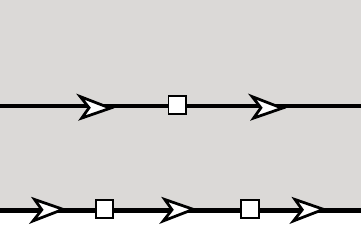}}}
   \put(0,15){
     \setlength{\unitlength}{.75pt}\put(0,0){
     \put(10,75)  {\scriptsize $\CM$}
     \put(120,20)  {\scriptsize $\CM$}
     \put(30,53) {\scriptsize $\CX$}
     \put(100,53) {\scriptsize $\CY$}
     \put(65,53){\scriptsize $\CW$}
     \put(63,-7) {\scriptsize $\CP$}
     \put(13,-7) {\scriptsize $\CP$}
     \put(115,-7) {\scriptsize $\CP$}
         
     }\setlength{\unitlength}{1pt}}
  \end{picture}} 
\quad\quad\quad\quad\quad\quad\quad\quad
\raisebox{-30pt}{
  \begin{picture}(120,85)
   \put(0,15){\scalebox{1}{\includegraphics{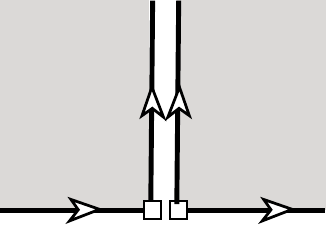}}}
   \put(0,15){
     \setlength{\unitlength}{.75pt}\put(0,0){
     \put(28,-8) {\scriptsize $\CP$}
     \put(100,-9) {\scriptsize $\CP$}
     
     \put(53,90) {\scriptsize $\CQ$}
     \put(67,90) {\scriptsize $\CQ'$}
     \put(10,50) {\scriptsize $\CM$}
     \put(105,50) {\scriptsize $\CM$}
         
     }\setlength{\unitlength}{1pt}}
  \end{picture}} 
$$
\caption{These pictures illustrate the physical meaning of the proof of Theorem$^{\mathrm{ph}}$\,\ref{pthm:cc-P}. 
}
\label{fig:cc-P}
\end{figure}

If we complete $\Sigma\CM$ by its condensation descendants, then we should also complete its boundary $\CP$ by including all condensation descendants and to form a $\Sigma\CM$-module. 
What is the condensation completion of $\CP$? It is clear that in order to include all possible defects on the boundary, we should not only include all possible 0d walls in $\CP$ but also those condensation descendants of $\CM$ fused into the boundary as illustrated in the first picture in Figure\,\ref{fig:cc-P} (a).

\begin{pthm} \label{pthm:cc-P}
The condensation completion of $\CP$ is given by $\Sigma\CP=\RMod_\CP(2\hilb)$, i.e. the 2-category of right $\CP$-modules in $2\hilb$.
\end{pthm}
\pf
Let $\CQ$ and $\CQ'$ be two gapped boundaries of the 2d topological order $(\CM,0)$ as illustrated in Figure\,\ref{fig:cc-P} (b). They are unitary multi-fusion categories. Notice that $(\CQ')^\rev \boxtimes_\CM \CP$ is an anomaly-free 1d topological order, i.e. a unitary multi-fusion category with trivial Drinfeld center \cite[Theorem\ 3.3.6]{kz}. Therefore, it evaporates into the trivial 2d topological order $2\hilb$. In this way, we obtain a new gapped boundary $\CQ$ of $\CM$ by fusing condensation descendants in $\CM$ to the boundary $\CP$. As a consequence, the condensation completion of $\CP$ realizes all possible gapped boundaries of $\CM$, each of which can be labeled by a right $\CP$-module $\CX$ (in $2\hilb$) and realized as $\fun_{\CP^\rev}(\CX,\CX)$. The unique 0d domain wall between two such boundaries $\fun_{\CP^\rev}(\CX,\CX)$ and $\fun_{\CP^\rev}(\CY,\CY)$ is given by $\fun_{\CP^\rev}(\CX,\CY)$ (or $\fun_{\CP^\rev}(\CY,\CX)$ depending on the orientation). We illustrate 0-,1-morphisms in 
the 2-category $\RMod_\CP(2\hilb)$ in the following diagram:
\begin{align} \label{quiver2}
\xymatrix{
\CX \ar@(ul,ur)[]^{\fun_{\CP^\rev}(\CX,\CX)} \ar@/^/[rr]^{\fun_{\CP^\rev}(\CX,\CY)} & & \CY \ar@(ul,ur)[]^{\fun_{\CP^\rev}(\CX,\CX)} \ar@/^/[ll]^{\fun_{\CP^\rev}(\CY,\CX)}
}.
\end{align}
We see immediately that the 0d,1d domain walls in the condensation completion of $\CP$ are precisely encoded by 1-,0-morphisms in $\RMod_\CP(2\hilb)$. Moreover, similar to (\ref{diag:vertical-fusion}), we see the fusion of 0d walls in $\Sigma\CP$ coincides with the composition of 1-morphisms in $\RMod_\CP(2\hilb)$. 
\epf

We denote the $E_0$-center of a unitary 2-category $\CS$ by $\FZ_0(\CS) \coloneqq \fun_{2\hilb}(\CS,\CS)$. 
\begin{prop}
For a unitary multi-fusion 1-category $\CP$, there is a braided equivalence 
\be \label{eq:center-1-center}
\FZ_1(\CP) \xrightarrow{\simeq} \Omega(\FZ_0(\Sigma\CP)) \quad\quad \mbox{defined by $(z,\beta_{z,-}) \mapsto \{ \CX \xrightarrow{z\odot -} \CX \}_{\CX\in\Sigma\CP}$.} 
\ee
\end{prop}
\pf
Notice that $\CX \xrightarrow{z \odot -} \CX$ is a well-defined $\CP$-module functor, which intertwines the $\CP$-action via the half-braiding $\beta_{z,-}$. This implies that (\ref{eq:center-1-center}) gives a well-defined functor. 
Conversely, there is a functor $ \Omega(\FZ_0(\Sigma\CP)) \to \FZ_1(\CP)$ defined by $\phi \mapsto \phi_\CP(\one_\CP)$. It is routine to check that it is well-defined and gives the quasi-inverse of (\ref{eq:center-1-center}). 
\epf

\begin{figure}
$$
\raisebox{-30pt}{
  \begin{picture}(130,90)
   \put(0,0){\scalebox{1}{\includegraphics{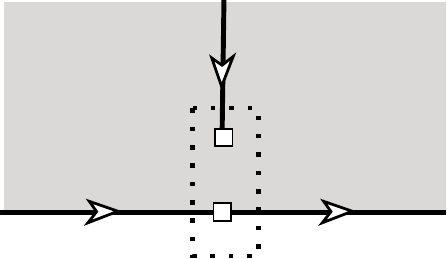}}}
   \put(0,-10){
     \setlength{\unitlength}{.75pt}\put(0,0){
     \put(83,48)  {\scriptsize $X$}
     \put(15,82)  {\scriptsize $\FZ_1(\CP)$}
     \put(135,82)  {\scriptsize $\FZ_1(\CP)$}
     \put(35,17) {\scriptsize $\CP$}
     \put(83,19) {\scriptsize $\CP$}
     \put(125,17) {\scriptsize $\CP$}
     \put(43,120) {\scriptsize $\fun_{\CP|\CP}(X\boxtimes_{\FZ_1(\CP)} \CP, X\boxtimes_{\FZ_1(\CP)} \CP)$}
     \put(65,5) {\scriptsize $X\boxtimes_{\FZ_1(\CP)} \CP$}
         
     }\setlength{\unitlength}{1pt}}
  \end{picture}} 
  \quad\quad\quad\quad\quad\quad\quad\quad\quad
  \raisebox{-30pt}{
  \begin{picture}(100,90)
   \put(0,0){\scalebox{1}{\includegraphics{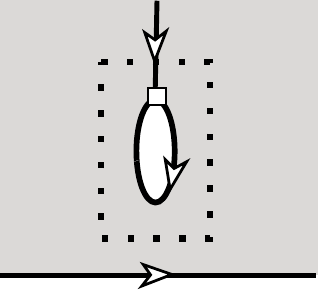}}}
   \put(0,-10){
     \setlength{\unitlength}{.75pt}\put(0,0){
     \put(67,85)  {\scriptsize $Y$}
     \put(3,110)  {\scriptsize $\FZ_1(\CP)$}
     \put(95,110)  {\scriptsize $\FZ_1(\CP)$}
     \put(67,45) {\scriptsize $\CP$}
     \put(125,17) {\scriptsize $\CP$}
     \put(13,130) {\scriptsize $\fun_{\FZ_1(\CP)^\rev}(Y\boxtimes_{\CP^\rev\boxtimes\CP} \CP, Y\boxtimes_{\CP^\rev\boxtimes\CP} \CP)$}
     \put(60,8) {\scriptsize $\CP$}
     \put(85,65) {\scriptsize $Y\boxtimes_{\CP^\rev\boxtimes\CP} \CP$}
         
     }\setlength{\unitlength}{1pt}}
  \end{picture}} 
$$
\caption{Based on the boundary-bulk relation \cite{kwz1}\cite[Theorem 3.3.7]{kz}, these pictures illustrate the physical meanings of the functor $\Phi$ and its inverse in the proof of Theorem$^{\mathrm{ph}}$\,\ref{thm:fun-PP}.}
\label{fig:fun-PP-2}
\end{figure}

\begin{pthm} \label{thm:fun-PP}
For a unitary multi-fusion 1-category $\CP$, there exists a natural monoidal equivalence: 
\be \label{eq:z0-sigma}
\Sigma\FZ_1(\CP) \simeq \FZ_0(\Sigma\CP) \coloneqq \fun_{2\hilb}(\Sigma\CP,\Sigma\CP). 
\ee
\end{pthm}
\begin{proof}
We have an evident equivalence $\FZ_0(\Sigma\CP) \simeq \BMod_{\CP|\CP}(2\hilb)$ and a monoidal functor $\Phi: \Sigma\FZ_1(\CP) \to \BMod_{\CP|\CP}(2\hilb)$ defined by $X\mapsto X \boxtimes_{\FZ_1(\CP)} \CP$ for objects and $f\mapsto f \boxtimes_{\FZ_1(\CP)} \id_\CP$ for 1-morphisms. It is routine to check that the assignment $Y \mapsto Y\boxtimes_{\CP^\rev\boxtimes\CP} \CP$ for objects and $g\mapsto \id_\CP\boxtimes_{\CP^\rev\boxtimes\CP}g$ for 1-morphisms defines an inverse of $\Phi$. The physical meanings of the construction of $\Phi$ and its quasi-inverse are illustrated in Figure\,\ref{fig:fun-PP-2}. 
\end{proof}

\begin{rem} \label{rem:fun-PP-proof-2}
Although the proof of Theorem$^{\mathrm{ph}}$\,\ref{thm:fun-PP} does have some physical meanings, it is not directly related to the condensation completion. We would like to sketch the physical meaning of Theorem$^{\mathrm{ph}}$\,\ref{thm:fun-PP} in terms of condensation completion. We denote the indecomposable $\CP$-modules in $2\hilb$ by $\CX_i$ for $i=1,\cdots, N$, and set $\CP_{\CX_i}^\vee \coloneqq \fun_{\CP^\rev}(\CX_i, \CX_i)$. To define a non-zero functor $F: \Sigma\CP \to \Sigma\CP$, we need to construct a non-zero monoidal functor $F_i:\CP_{\CX_i}^\vee \to \CP_{F(\CX_i)}^\vee$ for $i=1,\cdots, N$. By \cite[Theorem\ 3.2.3]{kz}, defining each $F_i$ amounts to choosing a closed multi-fusion $\FZ_1(\CP)$-$\FZ_1(\CP)$-bimodule $\FZ_1(\CP)_{\CY_i}^\vee \coloneqq \fun_{\FZ_1(\CP)^\rev}(\CY_i,\CY_i)$ for $\CY_i\in\Sigma\FZ_1(\CP)$ equipped with a monoidal equivalence $\FZ_1(\CP)_{\CY_i}^\vee\boxtimes_{\FZ_1(\CP)}  \CP_{\CX_i}^\vee \simeq \CP_{F(\CX_i)}^\vee$. Then $F_i$ can be constructed explicitly as the composed functor: 
\[
F_i: \CP_{\CX_i}^\vee \to \FZ_1(\CP)_{\CY_i}^\vee\boxtimes_{\FZ_1(\CP)}  \CP_{\CX_i}^\vee \xrightarrow{\simeq}\CP_{F(\CX_i)}^\vee.
\]
We illustrate the physical meaning of this construction in the following picture,
$$
\raisebox{-30pt}{
  \begin{picture}(130,90)
   \put(-40,15){\scalebox{1}{\includegraphics{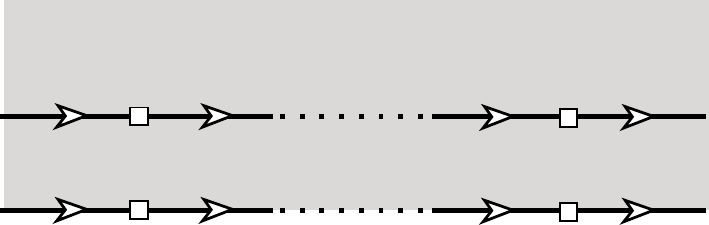}}}
   \put(-40,15){
     \setlength{\unitlength}{.75pt}\put(0,0){
     \put(183,52)  {\scriptsize $\CM_{\CY_{N-1}}^\vee$}
     \put(233,52)  {\scriptsize $\CM_{\CY_{N}}^\vee$}
     \put(20,52) {\scriptsize $\CM_{\CY_1}^\vee$}
     \put(81,52) {\scriptsize $\CM_{\CY_2}^\vee$}
     \put(20,-8) {\scriptsize $\CP_{\CX_1}^\vee$}
     \put(81,-7) {\scriptsize $\CP_{\CX_2}^\vee$}
     \put(185,-8){\scriptsize $\CP_{\CX_{N-1}}^\vee$}
     \put(240,-7) {\scriptsize $\CP_{\CX_{N}}^\vee$}
     \put(110,70){\scriptsize $\CM=\FZ_1(\CP)$}
     \put(45,50) {\scriptsize $(\CQ,q)$}
         
     }\setlength{\unitlength}{1pt}}
  \end{picture}} 
$$
where the 0d wall between $\CM_{\CY_1}^\vee$ and $\CM_{\CY_2}^\vee$ is given by a pair $(\CQ,q)$, where $\CQ=\fun_{\FZ_1(\CP)}(\CY_1,\CY_2)$ and $q\in\CQ$. Filling the 0d walls between $\CM_{\CY_i}^\vee$'s to give a construction of a functor $F$ is physically natural but not yet mathematically necessary at this stage (unless we apply Theorem$^{\mathrm{ph}}$\,\ref{thm:fun-PP}). Let us restrict the construction of $F$ to this type. Note that the data $q$ is needed to determine the functor $F$ on morphisms: 
\begin{align*}
\fun_{\CP^\rev}(\CX_i, \CX_j) &\xrightarrow{F} \CQ \boxtimes_{\FZ_1(\CP)} \fun_{\CP^\rev}(\CX_i, \CX_j) \simeq \fun_{\CP^\rev}(F(\CX_i), F(\CX_j))  \\
f &\mapsto q\boxtimes_{\FZ_1(\CP)} f.
\end{align*}
When $\CY_2=\CY_1$, it is necessary that $q=\one_{\CM_{\CY_1}^\vee}$. Moreover, $(\CQ,q)$ is required to be invertible. We spell out its precise meaning below. Let $(\CQ',q')$ be the 0d domain wall between $\CM_{\CY_2}^\vee$ and $\CM_{\CY_1}^\vee$. We must have
\begin{align*}
(\CQ,q)\boxtimes_{\CM_{\CY_2}^\vee} (\CQ',q') &= (\CQ\boxtimes_{\CM_{\CY_2}^\vee} \CQ', q\boxtimes_{\CM_{\CY_2}^\vee} q') \simeq (\CM_{\CY_1}^\vee, \one_{\CM_{\CY_1}^\vee}),   \\
(\CQ',q')\boxtimes_{\CM_{\CY_1}^\vee} (\CQ,q) &= (\CQ'\boxtimes_{\CM_{\CY_1}^\vee} \CQ, q'\boxtimes_{\CM_{\CY_1}^\vee} q) \simeq (\CM_{\CY_2}^\vee, \one_{\CM_{\CY_2}^\vee}). 
\end{align*}
As a consequence, the functor $\CM_{\CY_1}^\vee \to \CM_{\CY_2}^\vee$ defined by $f \mapsto q\boxtimes_{\CM_{\CY_2}^\vee} f \boxtimes_{\CM_{\CY_2}^\vee} q'$ is a monoidal equivalence. In other words, we have $\CM_{\CY_i}^\vee \simeq^\otimes \CM_{\CY_j}^\vee$ and $\CY_i\simeq\CY_j$. Then we can see that to define such a functor $F\in \fun_{2\hilb}(\Sigma\CP,\Sigma\CP)$ is equivalent to give an object $\CY$ in $\Sigma\FZ_1(\CP)$. 
\end{rem}

\begin{rem}
By Theorem${}^{\mathrm{ph}}$\,\ref{pthm:center-repM=trivial}, we obtain $\FZ_1(\FZ_0(\Sigma\CP))\simeq 2\hilb$, i.e. a typical example of the center of a center being trivial (recall the fact $\FZ_2(\FZ_1(\CP))\simeq 1\hilb$). This also means that $\Sigma\CP$ is a reasonable 2-codimensional description of an anomalous 1d topological order $\CP$. We expect this result and (\ref{eq:z0-sigma}) to hold for unitary multi-fusion $n$-categories. 
\end{rem}

\begin{rem} \label{rem:cc-gapless}
Conversely, if $\CS$ is a finite unitary 2-category such that $\FZ_1(\FZ_0(\CS))\simeq 2\hilb$, however, it is not true that $\CS \simeq \Sigma\CP$ for a unitary multi-fusion category $\CP$. Indeed, by By Theorem$^{\mathrm{ph}}$\,\ref{pthm:T=sigmaM}, $\FZ_0(\CS) \simeq \Sigma\CM$ for a UMTC $\CM$. As a consequence, $\CS$ describes a 1d boundary of the condensation completion of $\CM$. 
However, in general, $\CM$ is chiral, then this 1d boundary $\CS$ is necessarily gapless. By \cite{kz19a}, $\CS$ should be given by the condensation completion of a $\CB$-enriched unitary multi-fusion category ${}^\CB\CX$ for a UMTC $\CB$ and a unitary multi-fusion category $\CX$ such that $\FZ_1({}^\CB\CX)\simeq\CM$. Similar to Figure\,\ref{fig:cc-P}(b), we see that the condensation completion of ${}^\CB\CX$ produces all possible chiral gapless boundaries of $\CM$ and 0d walls between them. The precise mathematical description of this condensation completion of ${}^\CB\CX$  should be the delooping of ${}^\CB\CX$ in a proper sense. This involves some technical issues in the mathematical theory of enriched fusion categories that are beyond this work. We assume that this can be done, and denote this condensation completion by $\Sigma({}^\CB\CX)$. Then we obtain $\FZ_1(\FZ_0(\Sigma({}^\CB\CX)))\simeq 2\hilb$. Since this description of chiral gapless boundaries includes that of gapped boundaries as special cases, we obtain that $\CS\simeq\Sigma({}^\CB\CX)$ for a $\CB$-enriched unitary multi-fusion category ${}^\CB\CX$ such that $\FZ_1({}^\CB\CX)\simeq\CM$. 
\end{rem}

\subsection{2d SPT/SET orders} \label{sec:2d-SET-2}
The fact that the precise categorical description depends on the codimensions also applies to SET orders. Recall that an anomaly-free 2d SET order is described by two braided embeddings: 
$\CE\hookrightarrow \CC \hookrightarrow \CM$,
where all three unitary braided fusion 1-categories $\CE,\CC,\CM$ describe only particle-like excitations. They are clearly 0-codimensional descriptions. If we want to find a 1-codimensional description of an anomaly-free 2d SET order viewed as a boundary of the trivial 3d SPT order, then we need to do condensation completion. In particular, the condensation completions of $\CE$, $\CC$ and $\CM$ are given by fusion 2-categories $\Sigma\CE$, $\Sigma\CC$ and $\Sigma\CM$, respectively, as we explained in Remark\,\ref{rem:premodular}, and $\Sigma\CB\simeq\RMod_\CB(2\hilb)$ for any braided fusion category $\CB$.  


Since $\CE$ describes the bulk excitations of the trivial 2d SPT order, its condensation completion $\Sigma\CE$ describes all condensation descendants of symmetry charges in the trivial 2d SPT order. Its 3d bulk must be the trivial 3d SPT order, which is mathematically described by $\FZ_1(\Sigma\CE)$ \cite{kwz1,kwz2}. Since $\Sigma\CE$ is a symmetric fusion 2-category, we obtain a canonical braided embedding $\iota_0: \Sigma\CE \hookrightarrow \FZ_1(\Sigma\CE)$, which is necessarily full because $\Sigma\CE=2\Rep(G)$ or $2\Rep(G,z)$. Therefore, we obtain a pair $(\FZ_1(\Sigma\CE), \iota_0)$, which is precisely the 0-codimensional description of the trivial 3d SPT order (recall Theorem$^{\mathrm{ph}}$\,\ref{pthm:main-1}). 

Since $\CC$ describes the particle-like excitations of an anomaly-free 2d SET$_{/\CE}$ order, we should expect that the monoidal center of its condensation completion $\Sigma\CC$ should give again the trivial 3d SPT order. Mathematically, it means that there is a braided embedding $\iota: \Sigma\CE \hookrightarrow \FZ_1(\Sigma\CC)$ and a braided equivalence $\phi: \FZ_1(\Sigma\CE) \simeq \FZ_1(\Sigma\CC)$ rendering the following diagram commutative: 
\be \label{diag:rep-EC}
\raisebox{2em}{\xymatrix@R=1.5em{
& \Sigma\CE \ar@{^(->}[dl]_{\iota_0} \ar@{^(->}[dr]^{\iota} & \\
\FZ_1(\Sigma\CE) \ar[rr]_\simeq^\phi & & \FZ_1(\Sigma\CC)\, .
}}
\ee
We illustrate relations between $\Sigma\CE$ and $\Sigma\CC$ and their 3d bulks in Figure\,\ref{fig:2d-SET}. The braided equivalence $\phi$ is physically achieved by tunneling through the 2d invertible domain wall labeled by $\CY_\phi$, which is canonically associated to $\phi$. When $\CC=\CE$, $\CY_\phi$ is simply a 2d SPT order. Notice that this mathematical description is completely parallel to that of anomaly-free 1d SET orders (recall (\ref{diag:EA})). Therefore, we obtain that an anomaly-free 2d SET$_{/\CE}$ order is mathematically characterized by a pair $(\CC,\phi)$, where $\CC$ is a UMTC$_{/\CE}$ and $\phi: \FZ_1(\Sigma\CE) \to \FZ_1(\Sigma\CC)$ is a braided equivalence. 



\begin{rem}\label{rem:over-E=ff}
Since $\Sigma\CE$ includes symmetry charges and their condensation descendants, they can be moved in and out of the 2d boundary freely. Therefore, we expect that the composed functor $\Sigma\CE \stackrel{\iota}{\hookrightarrow} \FZ_1(\Sigma\CC) \to \Sigma\CC$ is monoidal and faithful. Mathematically, using $\Sigma\CE=\kar(B\CE)$ instead of $\RMod_\CE(2\hilb)$ (recall Remark\,\ref{rem:sigma-RMod}), one can show that there is a canonical faithful functor $\Sigma\CE \to \Sigma\CC$, which necessarily factors through $\FZ_1(\Sigma\CC) \to \Sigma\CC$. 
\end{rem}

\begin{rem}
It is not clear to us if the natural isomorphisms $\phi\circ\iota_0\simeq\iota$ rendering the diagram (\ref{diag:rep-EC}) commutative and the higher isomorphisms of them have any physical meanings. So we ignore them. See Remark\,\ref{rem:spt-groupoid} for the physical meanings of other higher isomorphisms. 
\end{rem}

When $\CC=\CE$, the pair $(\CE,\phi)$ gives a mathematical description of a 2d SPT order. Moreover, the stacking of SPT orders is described by the composition of autoequivalences. We summarize the classification result below.

We denote the underlying group of the category of braided autoequivalences of
$\FZ_1(\Sigma\CE)$ preserving $\iota_0$ by
$\Aut^{br}(\FZ_1(\Sigma\CE),\iota_0)$, and the underlying set of the category of braided equivalences from $(\FZ_1(\Sigma\CE),\iota_0)$ to $(\FZ_1(\Sigma\CC),\iota)$ by $\mathrm{BrEq}((\FZ_1(\Sigma\CE),\iota_0),(\FZ_1(\Sigma\CC),\iota))$, which is clearly an $\Aut^{br}(\FZ_1(\Sigma\CE),\iota_0)$-torsor. Moreover, an equivalence $\phi\in\Aut^{br}(\CC,\eta_\CC)$ of the 2d boundary theory should automatically induces an autoequivalence in $\Aut^{br}(\FZ_1(\Sigma\CC), \iota)$ of the 3d bulk. Hence, there is a natural $\Aut^{br}(\CC,\eta_\CC)$-action on $\mathrm{BrEq}((\FZ_1(\Sigma\CE),\iota_0),(\FZ_1(\Sigma\CC),\iota))$. 

\begin{pthm} \label{pthm:2d-SET-2}
The group of 2d SPT$_{/\CE}$ orders (with the multiplication defined by the stacking and the identity element defined by the trivial SPT$_{/\CE}$ order) is isomorphic to the group $\Aut^{br}(\FZ_1(\Sigma\CE),\iota_0)$. An anomaly-free 2d SET$_{/\CE}$ orders with particle-like excitations given by $(\CC,\eta_\CC)$ is characterized by $\phi\in\mathrm{BrEq}(\FZ_1(\Sigma\CE),\FZ_1(\Sigma\CC))$. Moreover, we have 
\[
\{ \mbox{2d anomaly-free SET$_{/\CE}$ orders with topological excitations $\CC$} \}  = \frac{\mathrm{BrEq}((\FZ_1(\Sigma\CE),\iota_0),(\FZ_1(\Sigma\CC),\iota))}{\Aut^{br}(\CC,\eta_\CC)}.
\]
\end{pthm}

\begin{figure}
$$
\raisebox{-30pt}{
  \begin{picture}(140,95)
   \put(0,15){\scalebox{1}{\includegraphics{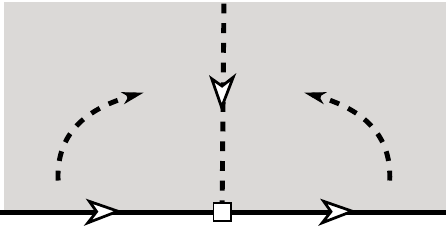}}}
   \put(0,15){
     \setlength{\unitlength}{.75pt}\put(0,0){
     \put(120,68)  {\scriptsize $(\FZ_1(\Sigma\CE),\iota_0)$}
     \put(5,68) {\scriptsize $(\FZ_1(\Sigma\CC),\iota)$}
     \put(35,-8) {\scriptsize $\Sigma\CC$}
     \put(125,-8){\scriptsize $\Sigma\CE$}
     \put(81,-7) {\scriptsize $\CX$}
     \put(83,93) {\scriptsize $\CY_\phi$}
         
     }\setlength{\unitlength}{1pt}}
  \end{picture}} 
  \hspace{3cm}
  \raisebox{-30pt}{
  \begin{picture}(20,95)
   \put(-10,15){\scalebox{1}{\includegraphics{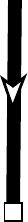}}}
   \put(-10,15){
     \setlength{\unitlength}{.75pt}\put(0,0){
     \put(1,-9) {\scriptsize $\CX$}
     \put(0,89) {\scriptsize $\Sigma\CM_\phi$}
         
     }\setlength{\unitlength}{1pt}}
  \end{picture}} 
$$
$$
\hspace{1.4cm} (a) \hspace{5.2cm} (b)
$$
\caption{Picture (a) depicts a physical configuration that illustrates the relation among $\Sigma\CE,\Sigma\CC$. 
$\CY_\phi$ denotes the 2d invertible domain wall between two 3d bulks associated to the braiding equivalence $\phi:\FZ_1(\Sigma\CE)\to \FZ_1(\Sigma\CC)$ and $\iota\simeq\phi\circ\iota_0$. Picture (b) depicts the result of the closing-fan process indicated by the dotted lines in Picture (a) (see (\ref{eq:sigma-M})).}
\label{fig:2d-SET}
\end{figure}

Theorem$^{\mathrm{ph}}$\,\ref{pthm:2d-SET-2} and its compatibility with the classification result in Theorem$^{\mathrm{ph}}$\,\ref{pthm:2d-SET-1} have many immediate mathematical consequences. In the rest of this subsection, we explain them in details and summarize the corresponding mathematical consequences as physical theorems. For mathematicians, these physical theorems should be viewed as mathematical conjectures. 

\medskip
If a UMTC$_{/\CE}$~$\CC$ is anomalous, on the one hand, its 3d bulk cannot be
the trivial 3d SPT order, hence, $\FZ_1(\Sigma\CE) \nsimeq \FZ_1(\Sigma\CC)$ as
braided fusion 2-categories; on the other hand, by \cite{LW160205936}, $\CC$
should not admit any minimal modular extension. These facts provide a
mathematical characterization of the existence of a minimal modular extension. 
\begin{pthm} \label{pthm:mext-center-EC}
A UMTC$_{/\CE}$~$\CC$ admits a minimal modular extension if and only if $\FZ_1(\Sigma\CE) \simeq \FZ_1(\Sigma\CC)$ as braided fusion 2-categories. 
\end{pthm}

\begin{rem}
It is important and interesting to provide a mathematical proof of Theorem$^{\mathrm{ph}}$\ \ref{pthm:mext-center-EC}. In \cite[Section\,V.C]{jf20}, one can find a more detailed discussion of the existence of minimal modular extensions.  
\end{rem}

If we close the fan depicted in Figure\,\ref{fig:2d-SET} to obtain an anomaly-free 2d topological order, the categorical description of such obtained anomaly-free 2d topological order is given by 
\be \label{EYC}
\Sigma\CC \boxtimes_{\FZ_1(\Sigma\CC)} \CY_\phi \boxtimes_{\FZ_1(\Sigma\CE)} (\Sigma\CE)^\rev.  
\ee
This anomaly-free 2d topological order provides a way of gauging $\Sigma\CC$ in
the same dimension. Therefore, it should come from a unique minimal modular
extension of the UMTC$_{/\CE}$~$\CC$. We denote this minimal modular extension by $\CM_\phi$. 

\begin{pthm}
We expect to have the following mathematical results.  
\bnu
\item There is a natural monoidal equivalence: 
\be \label{eq:sigma-M}
\Sigma\CC \boxtimes_{\FZ_1(\Sigma\CC)} \CY_\phi \boxtimes_{\FZ_1(\Sigma\CE)} (\Sigma\CE)^\rev \simeq \Sigma\CM_\phi. 
\ee
In particular, it implies that there is a braided equivalence 
\[
\Omega(\Sigma\CC \boxtimes_{\FZ_1(\Sigma\CC)} \CY_\phi \boxtimes_{\FZ_1(\Sigma\CE)} (\Sigma\CE)^\rev) \simeq \CM_\phi.
\] 
\item There is a bijection $h: \mathrm{BrEq}(\FZ_1(\Sigma\CE),\FZ_1(\Sigma\CC)) \to \mext(\CC)$ defined by 
\[
\phi \mapsto \Omega(\Sigma\CC \boxtimes_{\FZ_1(\Sigma\CC)} \CY_\phi \boxtimes_{\FZ_1(\Sigma\CE)} (\Sigma\CE)^\rev),
\]
intertwining the $\Aut^{br}(\CC,\eta_\CC)$-actions. When $\CC=\CE$, the map $h: \Aut^{br}(\FZ_1(\Sigma\CE),\iota_0) \to  \mext(\CE)$ is a group isomorphism. In particular, we should have the following explicit group isomorphisms. 
\[
\Aut^{br}(\FZ_1(2\Rep(G)),\iota_0) \simeq H^3(G,U(1)); \quad\quad \Aut^{br}(\FZ_1(2\Rep(\Zb_2,z)),\iota_0) \simeq \Zb_{16}; \quad\quad
\]
\enu
\end{pthm}

For a 2d SPT order, $\CC=\CE=\Rep(G)$, and for each $\phi\in\Aut^{br}(\FZ_1(2\Rep(G)),\iota_0)$, we have $\CM_\phi\simeq\FZ_1(1\hilb_G^\omega)$ for an $\omega\in H^3(G,U(1))$. It implies that the 1d defect junction $\CX$ in Figure\ \ref{fig:2d-SET}, as a 1d boundary of $\Sigma\CM_\phi$, is gapped. According to Theorem$^{\mathrm{ph}}$\,\ref{pthm:cc-P}, $\CX=\Sigma\CP$ for a unitary fusion 1-category $\CP$ equipped with a braided equivalence $\CM \to \FZ_1(\CP)$.
By the boundary-bulk relation, we obtain a monoidal equivalence $(\ref{EYC}) \simeq \fun_{2\hilb}(\CX,\CX)$. Hence, $\CY_\phi$ determines $\CX$ uniquely as the unique closed module over the unitary multi-fusion 2-category (\ref{EYC}). Therefore, we obtain a 2d generalization of (\ref{eq:pic-aut}) for $\CE=\Rep(G)$. 

\begin{pthm} \label{pthm:pic=aut-2d}
There are natural group isomorphisms: 
\begin{align*}
 \Pic(2\Rep(G)) &\simeq \Aut^{br}(\FZ_1(2\Rep(G)),\iota_0) \simeq H^3(G,U(1)).
\end{align*}
\end{pthm}

\void{
We would like to provide an independent mathematical proof of the following result. 
\begin{thm} \label{thm:pic-nrepG-coh}
$\Pic(2\Rep G) \simeq H^{3}(G,U(1))$. 
\end{thm}
\pf
By definition, $\Pic(2\Rep G)$ is the group formed by the equivalence classes of invertible objects of $3\Rep G$. That is, $\Pic(2\Rep G) = \pi_0(3\Rep G^\times)$, where $3\Rep G^\times$ is the $3$-groupoid obtained by discarding all non-invertible objects and morphisms of $3\Rep G$. We have $3\Rep G \simeq \kar(\fun(G,3\hilb))$, where $G$ is regarded as a groupoid with a single object, i.e. a 1-group. Thus $3\Rep G^\times \simeq \fun(G,3\hilb^\times)$. Note that the $3$-groupoid $3\hilb^\times$ consists of a single object, a single $k$-morphism for $1\leq k\le 2$ and the space of $3$-morphisms $\Cb^\times \simeq U(1)$. Namely, $3\hilb^\times \simeq K(\Zb,4)$. Therefore, $\Pic(2\Rep G) \simeq \pi_0 \fun(G,K(\Zb,4)) \simeq H^{4}(G,\Zb) \simeq H^{3}(G,U(1))$.
\epf
\begin{thm}
For a $2$-group $\CG^{(2)}$, $\Pic(\Rep(\CG^{(2)}))\simeq H^{3}(\CG^{(2)},U(1))$.
\end{thm}
\pf
The proof is the same as that of Theorem\,\ref{thm:pic-nrepG-coh}. 
\epf
}

\begin{rem} \label{rem:16-fold-way}
Interestingly, Theorem$^{\mathrm{ph}}$\,\ref{pthm:pic=aut-2d} does not hold for $\CE=\Rep(G,z)$. For example, when $G=\Zb_2$, for $\phi\in\Aut^{br}(\FZ_1(2\Rep(\Zb_2,z)),\iota_0)$ and $c_\phi\in\{ 0,\frac{1}{2}, \cdots, \frac{15}{2}\}$, the pairs $(\CM_\phi,c_\phi)$ reproduce the Kitaev’s 16-fold ways, which, viewed as 2d topological orders, all have chiral gapless edges except for the trivial one $\phi=\id_{\FZ_1(2\Rep(\Zb_2,z))}$. According to Remark\,\ref{rem:cc-gapless}, the condensation completion of a single chiral  gapless edge automatically includes all gapless edges of $(\CM_\phi,c_\phi)$. All of these gapless edges can be categorically described by a $\CB$-enriched multi-fusion category ${}^\CB\CX$, where $\CB$ is a UMTC that is Witt equivalent to $\CM_\phi$. In this case, we believe that it is possible to generalize the notion of invertible $2\Rep(\Zb_2,z)$-bimodules in an enriched setting. We denote the group of all such invertible bimodules by $\Pic^{\mathrm{en}}(2\Rep(\Zb_2,z))$. Since 16 UMTC's $\CM_\phi$ all have different Witt classes, we should expect that $\Pic^{\mathrm{en}}(2\Rep(\Zb_2,z))\simeq \Zb_{16}$ as groups. 
\end{rem}

Recall Example\,\ref{rem:n-SFC}, there are more unitary symmetric fusion $2$-categories than $2\Rep(G)$ and $2\Rep(G,z)$. For example, $2\hilb_H$ for a finite abelian group $H$. All unitary symmetric fusion $2$-categories should be viewed as certain higher symmetries. We would like to generalize Theorem$^{\mathrm{ph}}$\,\ref{pthm:2d-SET-2} to a higher symmetry defined by a unitary symmetric fusion 2-category $\CR$. 

\begin{defn}
A unitary fusion 2-category over $\CR$ is a unitary fusion 2-category $\CA$
equipped with a braided faithful functor $\iota_\CA: \CR \hookrightarrow
\FZ_1(\CA)$ such that the composed functor $\CR \hookrightarrow \FZ_1(\CA) \to \CA$ is faithful (recall Remark\,\ref{rem:over-E=ff}).
\end{defn}

The simplest example of a unitary fusion 2-category over $\CR$ is $\CR$. It is equipped with a canonical braided faithful functor $\iota_0: \CR \hookrightarrow \FZ_1(\CR)$. 

\begin{rem}
Notice that we do not require $\iota_\CA: \CR \hookrightarrow \FZ_1(\CA)$ to be full. This is because $\iota_0: \CR \hookrightarrow \FZ_1(\CR)$ is not full in general. For example, consider the fusion 2-category $2\hilb_H$ for a non-trivial finite abelian group $H$. It has a canonical symmetric fusion 2-category structure by choosing the trivial braidings and the trivial sylleptic structure \cite{C98}. We have $\FZ_1(2\hilb_H) \simeq \oplus_{h\in H} 2\Rep(H)$ as 2-categories \cite{ktz}. The canonical functor $2\hilb_H \to \FZ_1(2\hilb_H)$ is faithful but not full. 
\end{rem}

The following definition first appeared in \cite[Definition\,2.7]{dno} for 1-categories. 
\begin{defn}\label{def:equivalence-over-R}
An equivalence between two unitary fusion $2$-categories $\CA$ and $\CA'$ over $\CR$ is a monoidal equivalence $f:\CA\to \CA'$ rendering the following diagram commutative:
\[
\begin{tikzcd} 
\CR \arrow[hook]{r}   \arrow[hook]{d}  & \FZ_1(\CA) \arrow{d}{\simeq}[swap]{}   \\
\FZ_1(\CA') \arrow{r}{\simeq} & \FZ_1(f),
\end{tikzcd} 
\]
where $\FZ_1(f) = \fun_{\CA|\CA'}({}_f\CA', {}_f\CA')$ and ${}_f\CA'$ is an $\CA$-$\CA'$-bimodule with the left $\CA$-module structure induced from the monoidal functor $f$. 
\end{defn}

All such equivalences, together with higher isomorphisms, form a 2-groupoid. When $\CA'=\CA$, we denote the underlying group by $\Aut^\otimes(\CA,\iota_\CA)$. We denote the set of equivalence classes of braided equivalences $\phi:\FZ_1(\CR)\to \FZ_1(\CA)$ preserving the symmetries (i.e. $\phi\circ\iota_0\simeq\iota_\CA$) by $\mathrm{BrEq}((\FZ_1(\CR),\iota_0),(\FZ_1(\CA),\iota_\CA))$, which is equipped with a natural $\Aut^\otimes(\CA,\iota_\CA)$-action.

\begin{pthm} \label{pthm:2d-SET-R}
A $2$d SPT/SET order with a higher symmetry $\CR$ is called a $2$d SPT/SET$_{/\CR}$ order. We give the following classification. 
\bnu 
\item An anomaly-free 2d SET$_{/\CR}$ order is uniquely characterized by a pair $(\CA,\phi)$, where $\CA$ is a unitary fusion 2-category over $\CR$ describing all topological excitations (including all condensation descendants) and $\phi: \FZ_1(\CR) \to \FZ_1(\CA)$ is a braided equivalence rendering the following diagram commutative (up to natural isomorphisms): 
\[
\xymatrix@R=1.5em{
& \CR \ar@{^(->}[dl]_{\iota_0} \ar@{^(->}[dr]^{\iota_\CA} & \\
\FZ_1(\CR) \ar[rr]^\phi_\simeq & & \FZ_1(\CA).
}
\]

\item When $\CA=\CR$, the pair $(\CR, \phi)$ describes a 2d SPT$_{/\CR}$ order and $(\CR,\id_{\FZ_1(\CR)})$ is the trivial SPT order. Moreover, the group of all 2d SPT$_{/\CR}$ orders is isomorphic to the group $\Aut^{br}(\FZ_1(\CR),\iota_0)$.  For a given category of topological excitations (including all condensation descendants) $\CA$, i.e. a unitary fusion 2-category, we have
\enu
\[
\{ \mbox{2d anomaly-free SET$_{/\CR}$ orders with topological excitations $\CA$} \} = \frac{ \mathrm{BrEq}((\FZ_1(\CR),\iota_0),(\FZ_1(\CA),\iota_\CA))}{\Aut^\otimes(\CA,\iota_\CA)}\, .
\]
\end{pthm}


\begin{rem}
When $\CR=\Sigma\CE$ for $\CE=\Rep(G)$ or $\Rep(G,z)$, the unitary fusion $2$-category $\CA$ is necessarily monoidally equivalent to $\Sigma\CC$ for a unitary modular $1$-category $\CC$ over $\CE$ because all 1-codimensional defects are not detectable via braiding, thus must be condensation descendants.  
\end{rem}

\begin{figure}
$$
\raisebox{-30pt}{
  \begin{picture}(140,95)
   \put(0,15){\scalebox{1}{\includegraphics{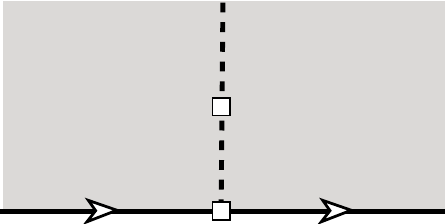}}}
   \put(0,15){
     \setlength{\unitlength}{.75pt}\put(0,0){
     \put(120,68)  {\scriptsize $(\FZ_1(\Sigma\CE),\iota_0)$}
     \put(5,68) {\scriptsize $(\FZ_1(\Sigma\CE),\iota_0)$}
     \put(35,-8) {\scriptsize $\Sigma\CE$}
     \put(125,-8){\scriptsize $\Sigma\CE$}
     \put(93,43) {\scriptsize $\CP$}
     \put(83,93) {\scriptsize $\CY_{\id}$}
         
     }\setlength{\unitlength}{1pt}}
  \end{picture}} 
$$
\caption{A 1d SPT $\CP$ is realized as an invertible domain wall on the trivial 2d SPT $\CY_{\id}$.}
\label{fig:2d-SPT}
\end{figure}

\begin{rem} \label{rem:spt-groupoid}
In Figure\,\ref{fig:2d-SPT}, if we insert an invertible gapped domain wall $\CP$ on the trivial 2d SPT order $\CY_{\id}$, one can see immediately that this 1d gapped domain wall is a 1d SPT order. Since an invertible domain wall defines a higher isomorphism of $\CY_{\id}$, by adding invertible domain walls of all higher codimensions on $\CY_\phi$ for all $\phi\in\Aut^{br}(\FZ_1(\Sigma\CE),\iota_0)$, we obtain a higher group of 2+1D SPT orders $\CG_3$. If $\CE=\Rep(G)$, by our classification of lower dimensional bosonic SPT orders, we obtain the formula $\pi_i(\CG_3) \simeq H^{3-i}(G,U(1))$, $i=0,1,2$. For a properly defined higher group of the minimal modular extensions of $\Rep(G)$, the same formula is rigorously proved (even for $i=3$, viewed as (-1+1)D SPT orders) in an upcoming paper \cite{DN20}. We can also consider domain walls between $\CY_\phi$ and $\CY_{\phi'}$ for $\phi\nsimeq\phi'$. These walls are gapless in general especially for the fermionic cases, i.e. $\CE=\Rep(G,z)$. They define even more interesting (i.e. gapless) higher morphisms. This picture clearly generalizes to all higher dimensional SPT/SET orders. As we have mentioned in the introduction, from a mathematical point of view, the right question is to study the category of SPT/SET orders instead of the set of them, but is beyond the scope of this work. In this work, we emphasize the importance of the boundary-bulk relation \cite{kwz2}, which should serve as the guiding principle for the future study of the category of SPT/SET orders as in \cite{kwz1}. 
\end{rem}

In the rest of this subsection, we discuss symmetry anomalies (recall Definition\,\ref{defn:hooft-anomaly}) of a  UMTC$_{/\CE}$ $(\CC,\eta_\CC)$, where $\CE$ is a unitary symmetric fusion 1-category. 

\smallskip
We first discuss a relation between anomalies and condensations. Microscopically, a condensation in a potentially anomalous phase can be achieved by introducing new interactions among excitations. Regarding this phase as the gapped boundary of a potentially non-trivial 1-dimensional-higher bulk, we see that these new interaction can be restricted to the boundary of the 1-dimensional-higher bulk without interacting with the physical degrees of freedom in the bulk. Although these new interactions might break the symmetries, it does not change the 1-dimensional-higher bulk up to different symmetries. Macroscopically, a condensation of a potentially anomalous phase produces a new potentially anomalous phase and a gapped domain wall between two phases, which share the same 1-dimensional-higher bulk. This physical intuition coincides with the precise mathematical theory of 1d condensations developed in \cite{kong}. More precisely, in 1d cases, 1d condensations produces Morita equivalent 1d phases, which share the same 2d bulk. We believe that the mathematical theory of higher dimensional SET orders and a higher Morita theory will eventually further confirm this physical intuition. We summarized this physical intuition in the following physical Theorem$^{\mathrm{ph}}$. 

\begin{pthm} \label{pthm:condensation-anomaly}
A (potentially symmetry-breaking) condensation of a (potentially anomalous) SET order does not change the gravitational anomaly (recall Definition\,\ref{defn:gravitational-anomaly}). When there is no gravitational anomaly, a condensation without breaking the symmetry preserves the symmetry anomaly. 
\end{pthm}

Next we argue that if $\mext(\CC,\eta_\CC)$ is empty, $(\CC,\eta_\CC)$ has no gravitational anomaly but only symmetry anomalies. Indeed, a breaking of symmetry can be achieved by condensing a condensable algebra $A$ in $\CE$. As a consequence, we obtain a UMTC $\CC_A$ over $\CE_A$ \cite{dno}. When the symmetry $\CE$ is bosonic, i.e. $\CE=\Rep(G)$, we can break the symmetry completely by choosing $A=\fun(G)$. In this case, $\CE_A=1\hilb$ and $\CC_A$ is a UMTC \cite[Corollary\ 4.31]{DGNO10}, which is anomaly-free. When the symmetry $\CE$ is fermionic, i.e. $\CE=\Rep(G,z)$, we can break the symmetry down to only the fermion parity symmetry $\Rep(\Zb_2, z)$ by choosing $A$ properly. This is always possible because there exists a canonical fiber functor $\Rep(G,z) \to \Rep(\Zb_2, z)$. It has long been conjectured that 2d fermionic topological orders are always anomaly-free \cite{BFHNPRW17}. 

By assuming above conjecture and by Definition\,\ref{defn:gravitational-anomaly} or Theorem$^{\mathrm{ph}}$\,\ref{pthm:condensation-anomaly}, we conclude that if the pair $(\CC,\eta_\CC)$ admits no minimal modular extension, it has no gravitational anomaly but only symmetry anomalies. In other words, its 3d bulk is a (potentially non-trivial) SPT order (recall Definition\,\ref{defn:hooft-anomaly}). A physical example was constructed in \cite{CBVF15}. As a consistent check, note that there is no 3d SPT order with only the fermion parity symmetry according to \cite[Table\ 2]{KTT1429} and \cite{FH160406527}. This result is consistent with the conjecture that 2d fermionic topological orders are all anomaly-free. By the classification of bosonic 3d SPT orders \cite{LW170404221,jf20}, when $\CE=\Rep(G)$, the SET order $(\CC,\eta_\CC)$, which admits no minimal modular extension, has only 't Hooft anomaly (recall Definition\,\ref{defn:hooft-anomaly}).

\medskip
More generally, let $\CR$ be a unitary symmetric fusion 2-category. 
\begin{pthm} \label{pthm:2d-thooft}
A 2d SET$_{/\CR}$ order with only a symmetry anomaly can be described by a unitary fusion 2-category over $\CR$, i.e. a pair $(\CA,\iota_\CA)$, together with a minimal modular extension $(\CM,\iota_\CM)$ of $\CR$ and a braided equivalence $\phi: \CM \to \FZ_1(\CA)$ satisfying $\iota_\CA \simeq \phi\circ \iota_\CM$. 
\end{pthm}

\begin{expl}
If $\CR=\Sigma\CE$ and $\CE=1\Rep(G)$, then we expect that $\CM\simeq^{br} \FZ_1(2\hilb_G^\omega)$ for a non-trivial cocycle $\omega\in H^4(G,U(1))$ (see Theorem$^{\mathrm{ph}}$\,\ref{pthm:main-1} and \cite{LW170404221,jf20}).  
\end{expl}

Recall that two UMTC$_{/\CE}$ $\CC_1$ and $\CC_2$ are called Witt$_{/\CE}$ equivalent if there are unitary fusion 1-categories $\CA_1,\CA_2$ over $\CE$ such that there is a braided equivalence 
\[
\CC_1 \boxtimes_\CE \FZ_2(\CE;\FZ_1(\CA_1)) \simeq \CC_2 \boxtimes_\CE \FZ_2(\CE;\FZ_1(\CA_2)).  
\]
The physical meaning of this Witt$_{/\CE}$ equivalence is that the 2d potentially anomalous SET's associated to $\CC_1$ and $\CC_2$ can be obtained via anyon condensations without breaking the symmetry from the same potentially anomalous SET \cite{dno}. By Theorem$^{\mathrm{ph}}$\,\ref{pthm:condensation-anomaly}, UMTC$_{/\CE}$'s in the same Witt$_{/\CE}$ equivalence class should share the same symmetry anomaly. As a consequence, there should be a well-defined group homomorphism 
\begin{align} \label{eq:witt}
\mbox{Witt$_{/\CE}$ group} &\xrightarrow{\FZ_1\circ \Sigma} \mext(\Sigma\CE)   \\
[\CC]_{/\CE} &\hspace{1mm} \mapsto \quad (\FZ_1(\Sigma\CC),\iota),  \nonumber
\end{align}
where $\iota$ is defined in Remark~\ref{rem:over-E=ff}. This group homomorphism might shed light on the study of both sides. We leave a systematic study of symmetry anomalies to the future. 

\subsection{\texorpdfstring{$n$}{n}d SPT/SET orders} \label{sec:nd-SET-2}

Theorem$^{\mathrm{ph}}$\,\ref{pthm:2d-SET-R} is ready to be generalized to all dimensions. For $n\geq 1$, let $\CR$ be a higher symmetry defined by a unitary symmetric fusion $n$-category $\CR$.

\void{
\medskip
We would like to reformulate above setup in a different but equivalent way, which is ready to be generalized. Set $\Sigma\CE \coloneqq \LMod_\CE((n+1)\vect)$. We introduce a couple of useful notions: 
\bnu
\item A fusion $n$-category $\CA$ is called connected if $\hom_\CA(\one_\CA,x)$ is not trivial for all non-zero object $x\in\CA$. It is not hard to see that a connected fusion $n$-category $\CA\simeq \Rep(\Omega\CA)$. 
\item A fusion $n$-category $\CA$ over a braided fusion $n$-category $\CB$ is a fusion $n$-category $\CA$ equipped with a braided embedding $\CB \hookrightarrow \FZ_1(\CA)$. For example, the braided fusion $n$-category $\CB$ is a fusion $n$-category over $\CB$ because there is a canonical embedding $\CB \hookrightarrow \FZ_1(\CB)$. 
\enu
}

\begin{defn} \label{def:n-fusion-over-R}
For $n\geq 1$, a unitary fusion $n$-category over $\CR$ is a pair $(\CA,\iota_\CA)$, where $\CA$ is a unitary fusion $n$-category $\CA$ and $\iota_\CA: \CR \hookrightarrow \FZ_1(\CA)$ is a braided faithful functor such that the composed functor $\CR \hookrightarrow \FZ_1(\CA) \to \CA$ is also faithful.
\end{defn}

The simplest unitary fusion $n$-category over $\CR$ is given by $(\CR,\iota_0)$, where $\iota_0: \CR \hookrightarrow \FZ_1(\CR)$ is the canonical braided faithful functor.  

\begin{defn}[\cite{dno} for $n=1$] \label{def:n-equivalence-over-R}
For $n\geq 1$, two unitary fusion $n$-categories over $\CR$, i.e. two pairs $(\CA,\iota_\CA)$ and $(\CA',\iota_{\CA'})$,  are called equivalent if there is a unitary monoidal equivalence $f:\CA\to \CA'$ rendering the following diagram commutative:
\[
\begin{tikzcd} 
\CR \arrow[hook]{r}{\iota_\CA}   \arrow[hook]{d}[swap]{\iota_{\CA'}}  & \FZ_1(\CA) \arrow{d}{\simeq}[swap]{}   \\
\FZ_1(\CA') \arrow{r}{\simeq} & \FZ_1(f),
\end{tikzcd} 
\]
where $\FZ_1(f) = \fun_{\CA|\CA'}({}_f\CA', {}_f\CA')$ and ${}_f\CA'$ is an $\CA$-$\CA'$-bimodule with the left $\CA$-module structure induced from the monoidal functor $f$. 
\end{defn}

All such equivalences, together with higher isomorphisms, form an $n$-groupoid. When $\CA'=\CA$, we denote the underlying group by $\Aut^\otimes(\CA,\iota_\CA)$. We denote the set of equivalence classes of braided equivalences $\phi:\FZ_1(\CR)\to \FZ_1(\CA)$ preserving the symmetries (i.e. $\phi\circ\iota_0\simeq\iota_\CA$) by $\mathrm{BrEq}((\FZ_1(\CR),\iota_0),(\FZ_1(\CA),\iota_\CA))$, which is equipped with a natural $\Aut^\otimes(\CA,\iota_\CA)$-action.

\begin{pthm} \label{pthm:main-2}
For $n\geq 1$, we call an $n$d (spatial dimension) SPT/SET order with a higher symmetry $\CR$ an $n$d SPT/SET$_{/\CR}$ order. We propose the following classification. 
\bnu 
\item An anomaly-free $n$d SET$_{/\CR}$ order is uniquely characterized by
a pair $(\CA,\phi)$, where $\CA$ is a unitary fusion $n$-category over $\CR$ describing all topological excitations (including all condensation descendants) and $\phi: \FZ_1(\CR) \to \FZ_1(\CA)$ is a braided equivalence rendering the following diagram commutative (up to natural isomorphisms):
\[
\xymatrix@R=1em{
& \CR \ar@{^(->}[dl]_{\iota_0} \ar@{^(->}[dr]^{\iota_\CA} & \\
\FZ_1(\CR) \ar[rr]_\simeq^\phi & & \FZ_1(\CA).
}
\]

\item When $\CA=\CR$, the pair $(\CR, \phi)$ describes an SPT$_{/\CR}$ order and $(\CR,\id_{\FZ_1(\CR)})$ describes the trivial SPT$_{/\CR}$ order. Moreover, the group of all SPT$_{/\CR}$ orders is isomorphic to the group $\Aut^{br}(\FZ_1(\CR),\iota_0)$, which denotes the underlying group of braided autoequivalences of $\FZ_1(\CR)$ preserving $\iota_0$. For a given category $\CA$ of topological excitations (including all condensation descendants), i.e. a unitary fusion $n$-category over $\CR$, we have
\enu
\[
\{ \mbox{$n$d anomaly-free SET$_{/\CR}$ orders with topological excitations $\CA$} \} = \frac{ \mathrm{BrEq}((\FZ_1(\CR),\iota_0),(\FZ_1(\CA),\iota_\CA))}{\Aut^\otimes(\CA,\iota_\CA)}\, .
\]
\end{pthm}

\begin{rem} \label{rem:TO}
When the higher symmetry $\CR$ is trivial, i.e. $\CR=n\hilb$, the statement of Theorem$^{\mathrm{ph}}$\,\ref{pthm:main-2} reduces to the classification of topological orders modulo invertible topological orders without symmetries (see \cite{kwz1,jf20}). This result was first proposed in \cite{kwz1}. But the definition of a multi-fusion $n$-category given in \cite{kwz1} is wrong due to the lack of Karoubi completion \cite{gjf19}. The correct one is given in \cite{jf20}, where one can also find many strong results with rigorous proof. The notion of unitarity for higher category is still missing (see Remark\,\ref{rem:unitary}). 
\end{rem}

\begin{rem} \label{rem:spatial-morita-eq}
Note that the braided equivalence $\phi: \FZ_1(\CR) \xrightarrow{\simeq} \FZ_1(\CA)$ seems to suggests that $\CA$ is ``Morita invertible'' over $\CR$. But one has to take this ``Morita invertibility'' with caution because the bimodule $\CX$ defining the ``Morita equivalence'' can be gapless (see the example given in Remark\,\ref{rem:16-fold-way}). An example of ``gapless Morita equivalence'' was introduced in \cite{kz19b} for 1+1D gapless edges of 2+1D topological orders under the name of ``spatially Morita equivalence''. 
\end{rem}

\begin{rem} The statement of Theorem$^{\mathrm{ph}}$\,\ref{pthm:main-2} makes sense if $\CR$ is only a unitary braided fusion $n$-category but not symmetric. In this case, $\CR$ contains non-trivial topological excitations. Can we still view $\CR$ as some kind of higher symmetries? It is interesting to investigate this question along the line of \cite{JW19}. 
\end{rem}

\void{
\begin{thm} \label{thm:pic-nrepG-coh}
$\Pic(n\Rep G) \simeq H^{n+1}(G,U(1))$. 
\end{thm}
\pf
By definition, $\Pic(n\Rep G)$ is the group formed by the equivlence classes of invertible objects of $(\nao)\Rep G$. That is, $\Pic(n\Rep G) = \pi_0((\nao)\Rep G^\times)$, where $(\nao)\Rep G^\times$ is the $(\nao)$-groupoid obtained by discarding all non-invertible objects and morphisms of $(\nao)\Rep G$. We have $(\nao)\Rep G \simeq \fun(G,(\nao)\hilb)$, where $G$ is regarded as a groupoid with a single object, i.e. a 1-group. Thus $(\nao)\Rep G^\times \simeq \fun(G,(\nao)\hilb^\times)$. Note that the $(\nao)$-groupoid $(\nao)\hilb^\times$ consists of a single object, a single $k$-morphism for $1\leq k\le n$ and the space of (\nao)-morphisms $\Cb^\times \simeq U(1)$. Namely, $(\nao)\hilb^\times \simeq K(\Zb,n+2)$. Therefore, $\Pic(n\Rep G) \simeq \pi_0 \fun(G,K(\Zb,n+2)) \simeq H^{n+2}(G,\Zb) \simeq H^{n+1}(G,U(1))$.
\epf

\begin{thm}
For an $n$-group $\CG^{(n)}$, $\Pic(\Rep(\CG^{(n)}))\simeq H^{n+1}(\CG^{(n)},U(1))$.
\end{thm}
\pf
The proof is the same as that of Theorem\,\ref{thm:pic-nrepG-coh}. 
\epf
}

\begin{rem} \label{rem:wen14}
We would like to point out that our classification of SPT orders, even in the case $\CR=n\Rep(G)$, goes beyond the usual group-cohomology classification (see also \cite{W1477}).

\void{where the $n$d bosonic SPT phases with symmetry $G$ are classified via cohomology $H^{n+1}(B(G \times SO_\infty),U(1))$. This classification is not one-to-one. Each element in $H^{n+1}(G \times SO_\infty,U(1))$ corresponds to a SPT phase, but different elements may sometimes correspond to the same SPT phase. For a finite group $G$, the above result can be reduced to $(\oplus_{k=1}^{n} H^k(G,i\mathrm{TO}^{n+1-k}))\oplus
H^{n+1}(G,U(1))$. Some examples are given in the following table: 
\[
\begin{tabular}{ |c|c|c|c|c|c|c|c| }
 \hline
symmetry & 0+1D & 1+1D & 2+1D &3+1D & 4+1D & 5+1D & 6+1D  \\
\hline
$Z_m$ & $\Zb_m$ & 0 & $\Zb_m$ & 0 & ${\Zb_m}\oplus\blue{\Zb_m}$ & $\blue{\Zb_{\<m,2\>}}$ & ${\Zb_m}\oplus\blue{\Zb_m\oplus \Zb_{\<m,2\>}}$  \\
$Z_2^T$ & 0 & $\Zb_2$ & 0 & ${\Zb_2}\oplus\blue{\Zb_2}$ & 0 & ${\Zb_2}\oplus \blue{2\Zb_2}$ & $\blue{\Zb_2}$ \\
\hline
\end{tabular}
\]
where entries in black are SPT phases in $H^{n+1}(G,U(1))$; those in blue are the beyond-group-cohomology SPT phases (i.e. in $\oplus_{k=1}^{n} H^k(BG,i\mathrm{TO}^{n+1-k})$). Here $i\mathrm{TO}^{n+1}$ is the abelian group of (\nao)D invertible topological orders, whose multiplication is stacking, and is given by
\begin{align*}
\begin{matrix}
n+1\text{D}: & 0+1 & 1+1 & 2+1 & 3+1 & 4+1 \\
i\mathrm{TO}^{n+1}: & 0 & 0 & \Zb & 0 & \Zb_2.\\
\end{matrix}
\end{align*}
The study of the precise relation between the results in this work and that in \cite{W1477} will be left to future works. }
\end{rem}

The assumed compatibility of Theorem$^{\mathrm{ph}}$\,\ref{pthm:main-1} and \ref{pthm:main-2} has some immediate consequences. We briefly discuss a few of them in the rest of this subsection. We start from the case $\CR=n\hilb$. In this case, we obtain the compatibility of 
the 0-codimensional description and the 1-codimensional description of anomaly-free $n$d topological orders for $n\geq 2$. This compatibility was a basic assumption in \cite{kwz1}.
\begin{pthm} 
For $n\geq 2$, a  unitary fusion $n$-category $\CT$ has the trivial monoidal center if and only if there is a unitary modular ($\nmo$)-category $\CM$ such that $\CT \simeq \Sigma\CM$ as unitary fusion $n$-categories. 
\end{pthm}

\begin{rem} \label{rem:theo-n-umtc}
The non-unitary version of this result is proved recently by Johnson-Freyd in \cite[Corollary IV.2]{jf20}. Although Johnson-Freyd ignored the subtle issue of the lack of a universally accepted and well developed model or theory of weak $n$-categories, for physically oriented readers, we believe that it is safe to take this result as a theorem. 
\end{rem}

Similarly, we expect a generalization of (\ref{eq:z0-sigma}) for higher fusion categories. 
\begin{pthm}
For $n\geq 1$ and a unitary fusion $n$-categories $\CP$, there should be a natural monoidal equivalence: 
\be \label{eq:z0-sigma-2}
\Sigma\FZ_1(\CP) \simeq \FZ_0(\Sigma\CP) \coloneqq \fun_{n\hilb}(\Sigma\CP,\Sigma\CP). 
\ee
\end{pthm}

When the higher symmetry $\CR$ is non-trivial, the assumed compatibility of Theorem$^{\mathrm{ph}}$\,\ref{pthm:main-1} and \ref{pthm:main-2} can be stated more precisely. Let $\CE$ be a unitary symmetric fusion $n$-category and $(\CC,\eta_\CC)$ a unitary modular $n$-category over $\CE$. Similar to the 2d case (recall Remark\,\ref{rem:over-E=ff}), we expect that there should be a natural braided embedding $\iota: \Sigma\CE \hookrightarrow \FZ_1(\Sigma\CC)$ rendering $\Sigma\CC$ a unitary fusion ($\nao$)-category over $\Sigma\CE$.  
\begin{pthm}
There should be a natural isomorphism between the set of minimal modular extensions of $(\CC,\eta_\CC)$ and the set of $\mathrm{BrEq}((\FZ_1(\Sigma\CE),\iota_0), (\FZ_1(\Sigma\CC),\iota))$. When $\CC=\CE$, this isomorphism should be a group isomorphism. 
\end{pthm}

The discussion of $1$d, $2$d SET$_{/\CR}$ orders with only symmetry anomalies (recall Theorem$^{\mathrm{ph}}$\,\ref{pthm:1d-thooft} and \ref{pthm:2d-thooft}) can be generalized to $n$d cases directly. 
\begin{pthm} \label{pthm:nd-t-hooft-anomaly}
For $n\geq 1$, an $n$d (spatial dimension) SET$_{/\CR}$ order with only a symmetry anomaly (without gravitational anomaly) is characterized by a quintuple $(\CA,\iota_\CA;\CM,\iota_\CM;\phi)$, where $(\CA,\iota_\CA)$ is a unitary fusion $n$-category over $\CR$ describing all topological excitations (including all condensation descendants), $(\CM,\iota_\CM)$ is a minimal modular extension of $(\CR,\id_\CR)$ (which determines the symmetry anomaly), and $\phi: \CM \to \FZ_1(\CA)$ is a braided equivalence rendering the following diagram commutative (up to natural isomorphisms): \[ \xymatrix@R=1em{ & \CR \ar@{^(->}[dl]_{\iota_\CM} \ar@{^(->}[dr]^{\iota_\CA} & \\ \CM \ar[rr]_\simeq^\phi & & \FZ_1(\CA).}
\] 
If $(\CM,\iota_\CM)\nsimeq(\FZ_1(\CR),\iota_0)$, then the symmetry anomaly is non-trivial. We obtain 
\begin{align*}
&\{ \mbox{$n$d SET$_{/\CR}$ orders with symmetry anomaly $(\CM,\iota_\CM)$ and topological excitations $\CA$} \}  \\
&\hspace{2cm} = \frac{ \mathrm{BrEq}((\CM,\iota_\CM),(\FZ_1(\CA),\iota_\CA))}{\Aut^\otimes(\CA,\iota_\CA)}\, .
\end{align*}
\end{pthm}

\begin{rem}
Since $(\CM,\iota_\CM)$ describes an $\nao$d SPT order which is invertible under stacking, we expect a natural group isomorphism $\Aut^{br}(\CM,\iota_\CM) \simeq \Aut^{br} (\FZ_1(\CR),\iota_0)$. Note that all $n$d invertible domain walls associated to $\phi\in\Aut^{br}(\CM,\iota_\CM)$ can be obtained from stacking $n$d SPT orders with the trivial domain wall associated to $\id_\CM$. 
\end{rem}

\begin{rem}
For our classification results to be useful in the study of real cases in practice, the key is to compute the monoidal center of higher fusion $n$-categories. Unfortunately, not many results on this problem are available partially because higher category theory is still underdeveloped and partially because computing center is already non-trivial for fusion 1-, 2-categories. As far as we know, only the categories of topological excitations of Dijkgraaf-Witten theories (as the monoidal centers of $n\mathrm{Vec}_G^\omega$) were computed in \cite{will08} for the 2d cases and in \cite{ktz} for the 3d cases, and conjectured in \cite{ktz} for all higher dimensional cases. In a unique situation, when the 1-dimensional higher bulk is known to be 3d and bosonic, there is no need to compute the center because this bulk is uniquely determined by a 4-cocycle in $H^4(G,U(1))$ according to \cite{LW170404221}. In this case, one can treat the 4-cocycle directly as the anomaly as in \cite{WLL16,BC20,BB20}.  
\end{rem}

\begin{rem}\label{anoset}
The generalization to mixed gravitational and symmetry anomaly is immediate. In short, one simply allow the modular extension $\CM$ in the above theorem to be \emph{not minimal}. Note that the minimal modular extension description of an anomaly-free $\nao$d SET$_{/\CR}$ order may be simplified to the data $\CR\xrightarrow{\iota_\CM} \CM$, where $\CM$ is a unitary modular $n$-category and $\iota_\CM$ is a braided embedding. The topological excitations are given by the $E_2$-centralizer $\FZ_2(\iota_\CM)$ (recall Remark\,\ref{rem:E2-center}), and $\CM$ is naturally a minimal modular extension of $\FZ_2(\iota_\CM)$. The topological excitations of an anomalous $n$d SET order are still described by a fusion $n$-category $\CA$ over $\CR$. By the boundary-bulk relation, we obtain the following commutative diagram: 
  \be 
  \raisebox{2em}{\xymatrix@R=1em{
      & \CR \ar@{^(->}[dl]_{\iota_\CM} \ar@{^(->}[dr]^{\iota_\CA} & \\
	  \CM \ar[rr]_\simeq^\phi & & \FZ_1(\CA)\, ,
      }}
  \ee
  which automatically induces an equivalence between the bulk excitations
  $\FZ_2(\iota_\CM)\simeq \FZ_2(\iota_\CA)$. When $\FZ_2(\iota_\CM)$ is larger
  than $\CR$, there is a gravitational anomaly.
\end{rem}

\begin{rem} \label{rem:tower}
We want to emphasize again that the second approach based on the idea of boundary-bulk relation is not independent from the first approach based on the idea of gauging the symmetry. More precisely, the categorical description of a potentially anomalous $n$d SPT/SET order depends on the categorical description of 1-higher-dimensional SPT orders obtained from the idea of gauging the symmetry. If we do not gauge the symmetry of the 1-higher-dimensional SPT order, one can still apply the boundary-bulk relation (by taking an over-$\CR$ center), but then the categorical data of a 1-dimension-higher SPT order is not complete either. One should consider even higher dimensional SPT orders. As a consequence, we obtain an infinite tower of higher dimensional SPT orders. Although one can speculate if this infinite tower can provide a precise description of an $n$d SPT/SET order, one can see, from this perspective, that applying the trick of gauging the symmetry allows us to truncate the tower and complete the missing data in the 1-higher dimension. In retrospective and from a mathematical point of view, to be able to gauge the symmetry in the same dimension is both miraculous and mysterious, and demands further studies. 
\end{rem}

\void{
\begin{table*}[t] \label{table:iTO-wen}
 \caption{
Entries in black are SPT phases described by group cohomology $H^{d+1}(G,U(1))$. The results in blue are the beyond-group-cohomology SPT phases described by $\oplus_{k=1}^{d} H^k(BG,i\mathrm{TO}^{d+1-k})$.  
} 
\label{SPT}
 \centering
 \begin{tabular}{ |c|c|c|c|c|c|c|c| }
 \hline
symmetry & 0+1D & 1+1D & 2+1D &3+1D & 4+1D & 5+1D & 6+1D  \\
\hline
$Z_n$ & $\Zb_n$ & 0 & $\Zb_n$ & 0 & ${\Zb_n}\oplus\blue{\Zb_n}$ & $\blue{\Zb_{\<n,2\>}}$ & ${\Zb_n}\oplus\blue{\Zb_n\oplus \Zb_{\<n,2\>}}$  \\
$Z_2^T$ & 0 & $\Zb_2$ & 0 & ${\Zb_2}\oplus\blue{\Zb_2}$ & 0 & ${\Zb_2}\oplus \blue{2\Zb_2}$ & $\blue{\Zb_2}$ \\
\hline
 \end{tabular}
\end{table*}
}


\appendix

\section{Appendix} \label{sec:appendix}

In this Appendix, we briefly review Gaiotto and Johnson-Freyd's construction of the Karoubi completion of an $n$-category \cite{gjf19} and the notion of a (braided) multi-fusion $n$-category \cite{jf20} and various higher centers \cite{lurie}. 


\medskip
Consider two potentially anomalous $n$d gapped phases $X$ and $Y$. A condensation of $X$ onto $Y$ is a pair of $(\nmo)$d gapped domain walls $f:X\rightleftarrows Y:g$, together with a condensation of the composite $(\nmo)$d gapped wall $f\circ g$ from $Y$ to $Y$ onto the trivial wall $\id_Y$ from $Y$ to $Y$. This leads to the precise formulation of an $n$-condensation between two objects $X$ and $Y$ in an $n$-category (assumed to be weak) defined inductively. More precisely, 0-condensations are equalities among elements in a set. An $n$-condensation of $X$ onto $Y$ is pair of 1-morphisms $f:X\rightleftarrows Y:g$, together with an $(\nmo)$-condensation of $fg$ to $\id_Y$. A walking $n$-condensation is the $n$-category $\spadesuit_n$ freely generated by an $n$-condensation. Therefore, an $n$-condensation in $\CC$ is precisely a functor $\spadesuit_n \to \CC$. Its full subcategory on the object $X$ is denoted by $\clubsuit_n$. A condensation $n$-monad in an $n$-category $\CC$ is a functor $\clubsuit_n \to \CC$. An $n$-category $\CC$ has all condensates if all $(\nmo)$-categories $\Hom_\CC(X,Y)$ have all condensates and every condensation monad $\clubsuit_n \to \CC$ extends to a condensation $\spadesuit_n \to \CC$, which is automatically unique if exists. An $n$-category is called Karoubi complete if every condensation monad factors through a condensation. Moreover, by \cite[Theorem\ 2.3.10]{gjf19}, for an $n$-category $\CC$ whose $(\nmo)$-categories of morphisms have all condensates, there is an $n$-category $\kar(\CC)$ called the ``Karoubi envelope'' (or ``Karoubi completion'') of $\CC$, such that there is a fully faithful functor $\CC \to \kar(\CC)$, which is an equivalence if $\CC$ has all condensations. 

An $n$-category is called $\Cb$-linear if the set of $n$-morphisms of given domain and codomain is a vector space over $\Cb$ and all compositions are $\Cb$-linear in each variables; called additive if the direct sum is defined for all $k$-morphisms $0\leq k<n$. Following \cite{jf20}, we define ``the delooping'' of a $\Cb$-linear additive Karoubi complete $n$-category $\CC$ by $\Sigma\CC \coloneqq \kar(B\CC)$, where $B\CC$ denotes the ``one-point delooping'' of $\CC$. We denote the category of $\Cb$-linear additive Karoubi complete monoidal $n$-categories (as objects) and bimodules (as 1-morphisms) by $\Alg_{E_1}^{\mathrm{Mor}}(n\catkc)$, which is itself symmetric monoidal with tensor product $\boxtimes$ defined by the naive tensor product $\otimes$ followed by a Karoubi completion \cite[Section\ II.B]{jf20}.  
\begin{defn}[\cite{jf20}] \label{def:jf}
A monoidal $n$-category $\CC$ is multi-fusion if it is additive, $\Cb$-linear, Karoubi complete and fully dualizable in $\Alg_{E_1}^{\mathrm{Mor}}(n\catkc)$. It is fusion, if, in addition, $\Omega^n\CC=\Cb$. 
\end{defn}

A braided monoidal $n$-category is an $E_2$-algebra in the $(\nao)$-category of $n$-categories, or equivalently, an $E_1$-algebra in the $(\nao)$-category of monoidal $n$-categories. A symmetric monoidal $n$-category is an $E_{n+2}$-algebra (automatically an $E_\infty$-algebra) in the $(\nao)$-category of $n$-categories. We assume that a proper notion of a unitary braided (multi-)fusion $n$-category can be defined. 
\begin{rem} \label{rem:unitary}
Tentatively, following \cite[Definition A.4]{kwz1}), a unitary $n$-category is a $\Cb$-linear category $\CC$ equipped with an equivalence $\delta: \CC \to \CC^\op$ fixing all $k$-morphisms for $0\leq k<n$, and is antilinear, involutive and positive on $n$-morphisms, i.e.
\[
\delta(\lambda f) \simeq \bar{\lambda} \delta(f), \quad\quad \delta\delta(f)=f, \quad\quad f\circ \delta(f) = 0 \Rightarrow f=0 ,
\]
for all $n$-morphisms $f$ in $\CC$ and $\lambda\in\Cb$. If $\CC$ has adjoints (i.e. all $k$-morphisms have the left and right adjoints for $1\leq k<n$), then the left adjoint and the right adjoint of a $k$-morphism $f$ are canonically equivalent \cite[Proposition\ A.7]{kwz1}. An $n$-functor $F: \CC\to\CC'$ is unitary if $F$ is $\Cb$-linear for $n$-morphisms and $F\circ \delta = \delta \circ F$. 
\end{rem}

Using the definition of $\Sigma\CC \coloneqq \kar(B\CC)$, one can define $\Sigma\Cb \coloneqq \vect$, i.e. the category of finite dimensional vector spaces over $\Cb$, and $n\vect \coloneqq \Sigma^n\Cb$. We assume that it generalizes to the unitary cases, $\Sigma\Cb=\hilb$ and $n\hilb \coloneqq \Sigma^n\Cb$. By \cite[Corollary 4.2.3 \& 4.2.4]{gjf19}, for a multi-fusion $n$-category $\CA$, $\Sigma\CA=\kar(B\CA)$ is equivalent to the category $\RMod_\CA^{\mathrm{fd}}$ of fully dualizable $\CA$-module $n$-categories. We further assume that the Karoubi completion is compatible with the notion of unitarity. For a unitary multi-fusion $n$-category, we expect that there is an equivalence $\Sigma\CA \simeq \RMod_\CA^{\mathrm{fd}}((\nao)\hilb)$.

\medskip
For a multi-fusion category $\CC$, its monoidal center or $E_1$-center, denoted by $\FZ_1(\CC)$, can be defined via the universal property \cite{lurie}, or more concretely by $\fun_{\CC\boxtimes\CC^\op}(\CC,\CC)$. The notion of an $E_n$-center, denoted by $\FZ_n(-)$, can also be defined by the universal property, or more concretely by the $E_1$-center of an $E_1$-algebra in the category of $E_{n-1}$-algebras \cite{lurie}. The notion of centralizers can be defined similarly (see \cite[Section\ 5.3]{lurie}). A unitary modular $n$-category can be defined as a unitary braided fusion $n$-category $\CC$ with a trivial $E_2$-center (or satisfying $\FZ_1(\CC) \simeq \CC\boxtimes\overline{\CC}$).

\end{document}